%% file: document.tex
\documentclass[onecolumn]{article}

\usepackage[a4paper, total={6in, 8.5in}]{geometry}
\usepackage{amsmath, amsfonts, amssymb, amsthm}
\usepackage{graphicx}
\usepackage{authblk}
\usepackage{titlesec}
\titleformat*{\section}{\large\bfseries}
\titleformat*{\subsection}{\large\bfseries}

\usepackage{algorithm}
\usepackage{algorithmic}
\usepackage{amsthm}
\usepackage{subfigure}
\usepackage{epsfig}
\usepackage{graphicx}
\usepackage{url}
\usepackage{multirow}
\usepackage{amssymb}
\usepackage{amsthm}
\usepackage{mathrsfs}
\usepackage{epstopdf}
\usepackage{mathtools}
\usepackage{multicol}

\DeclarePairedDelimiter\floor{\lfloor}{\rfloor}

\newtheorem{theorem}{\textbf{Theorem}}
\newtheorem{assumption}{\textbf{Assumption}}
\newtheorem{lemma}{\textbf{Lemma}}
\newtheorem{corollary}{\textbf{Corollary}}
\newtheorem{remark}{\textbf{Remark}}
\newtheorem{definition}{\textbf{Definition}}

\newcommand\blfootnote[1]{%
	\begingroup
	\renewcommand\thefootnote{}\footnote{#1}%
	\addtocounter{footnote}{-1}%
	\endgroup
}

\title{Adaptive Serverless Learning}
\author[1]{Hongchang Gao}
\author[2]{Heng Huang}

\affil[1]{Temple University}
\affil[2]{University of Pittsburgh}

\date{}

\begin{document}

\maketitle

\begin{abstract}
	With the emergence of distributed data, training machine learning models in the serverless manner has attracted increasing attention in recent years. Numerous training approaches have been proposed in this regime, such as decentralized SGD. However, all existing decentralized algorithms only focus on standard SGD. It might not be suitable for some applications, such as deep factorization machine in which the feature is highly sparse and categorical so that the adaptive training algorithm is needed. 
	In this paper, we propose a novel adaptive decentralized training approach, which can compute the learning rate from data dynamically.  To the best of our knowledge, this is the first adaptive decentralized training approach. Our theoretical results reveal that the proposed algorithm can achieve linear speedup with respect to the number of workers. Moreover, to reduce the communication-efficient overhead, we further propose a communication-efficient adaptive decentralized training approach, which can also achieve linear speedup with respect to the number of workers. At last, extensive experiments on different tasks have confirmed the effectiveness of our proposed two approaches. 
\end{abstract}

\blfootnote{hongchanggao@gmail.com, heng.huang@pitt.edu}

\section{Introduction}
\vspace{-5pt}
With the development of intelligence devices, a huge amount of data is generated and distributed in diverse devices. The emergence of the huge amount of distributed data challenges the feasibility of traditional training strategies for machine learning models. In recent years, to address this issue, distributed training approaches for large-scale machine learning models have attracted a surge of attention in machine learning community.  In particular, a distributed training system contains multiple worker nodes, which can simultaneously train the machine learning model on the local data of each worker node.  In this way, the large-scale distributed data can be efficiently handled by the computational power of these worker nodes. Formally, a distributed training system is to optimize the following problem:
\begin{equation} \label{lossfunction}
\begin{aligned}
&    \min_{\mathbf{x}} f(\mathbf{x}) = \frac{1}{K}\sum_{k=1}^K f^{(k)}(\mathbf{x}) \ , \\
\end{aligned}
\end{equation}
where $\mathbf{x}\in \mathbb{R}^d$ denotes the model parameter and $K$ indicates the number of worker nodes. Here, the $k$-th worker node optimizes the loss function $f^{(k)}(\mathbf{x})=\mathbb{E}_{\xi\sim\mathcal{D}^{(k)}} F^{(k)}(\mathbf{x}; \xi)$ where $\mathcal{D}^{(k)}$ is the data distribution on the $k$-th worker. In this paper,  we restrict our focus on the non-convex problem.

A commonly used distributed training approach to solve Eq.~(\ref{lossfunction}) is the centralized parallel stochastic gradient descent (C-PSGD).  There are one server node and multiple worker nodes. In each iteration, the worker node computes the gradient based on its local data and sends it to the server node, then the server node averages these gradients and sends the averaged one back to each worker node. It can be seen that all worker nodes need to communicate with the server node, which might cause the communication traffic jam on the server node and downgrade the parallelization performance. An alternative strategy is the serverless architecture (a.k.a. decentralized training), where each worker node updates the model parameter based on its local data and only communicates with its neighbors rather than the central server. In other words, the decentralized training approach is more friendly to communication, especially when the number of worker nodes is very large. As a result, the decentralized training approach has attracted more and more attention recently. 

Considering the promising performance in communication, a lot of decentralized training algorithms \cite{alghunaim2019linearly,bianchi2013performance,lian2017can,reisizadeh2019robust,seaman2017optimal,yu2019linear} have been proposed in recent years. For instance, \cite{lian2017can} proposed the decentralized parallel stochastic gradient descent (D-PSGD) algorithm based on the gossip averaging method and showed that the convergence rate of D-PSGD has the consistent leading term with that of C-PSGD.  However, all these approaches only focus on decentralizing standard SGD, where the same learning rate is employed across different worker nodes. It's not satisfactory for many real-world applications. For instance, \cite{zhang2019adam} discloses that standard SGD is not suitable for the case where the stochastic gradient has the heavy-tailed noise. In addition, the data on different worker nodes may have different properties so that it is not reasonable to use the same learning rate across different worker nodes.  Moreover, it is impractical to tune the learning rate for each worker in a decentralized training system. Therefore, it is necessary and important to incorporate the adaptive learning rate for decentralized training. However, it is challenging and unclear how to apply the adaptive learning rate to the decentralized training approach with convergence guarantee due to its considerable complication. In particular, adopting adaptive learning rates will cause different worker nodes have different learning rates so that it is much more difficult to analyze its convergence rate.

In this paper, to address the aforementioned issues, we propose a new decentralized training approach with adaptive learning rates: decentralized Adam. In detail, each worker node employs its own learning rate in terms of the local data as Adam~\cite{kingma2014adam}. Consequently, the learning rate is adapted to the specific data on each worker node, and there is no need to tune the learning rate for each worker node.   As far as we know, this is the first \textit{decentralized} training approach with adaptive learning rates. Furthermore, unlike most existing decentralized training approaches that conduct communication at each iteration, our approach conducts communication at every $p$ (where $p>1$) iterations to reduce the communication overhead. To the best of our knowledge, this is also the first \textit{adaptive decentralized} training approach with skipping communication. However, skipping communication rounds leads to new challenges for the convergence analysis, especially under the setting of adaptive learning rates. In this paper, we successfully address these challenging issues and obtain the convergence rate of our proposed approach. In particular, the linear speedup with respect to the number of workers is established, which is the first work obtaining this conclusion in this regime. 

However, although our proposed first approach can reduce the number of communication rounds, yet the communication cost in each round can still be the bottleneck when the model is large. For instance, a regular ResNet152 \cite{he2016deep} model is as large as 240MB. To mitigate this issue, we further propose a  communication-efficient adaptive decentralized training approach, which compresses the communicated model parameter between neighbor nodes. Although compressed communication has been studied in the decentralized optimization area \cite{koloskova2019decentralizednonconv,koloskova2019decentralizedcon,tang2018communication}, yet all existing approaches only focus on standard SGD. They cannot be directly applied to the adaptive learning rate case. In fact, the adaptive learning rate and skipping communication rounds cause new challenges for its convergence analysis. In this paper, we also proved its convergence rate and  the linear speedup. At last, our extensive experimental results have verified the effectiveness of our proposed two approaches.  Here, we summarize the contributions of our work as follows:
\begin{itemize}%[leftmargin=0.36in]
	\setlength{\itemsep}{0pt}
	\vspace{-3pt}
	\item We propose a new adaptive decentralized training approach, which employs the adaptive learning rate on each worker node.  This is the first work that employs adaptive learning rates for decentralized training. 
	\item We also propose a new communication-efficient adaptive decentralized training approach, which skips communication rounds and compresses the communicated model parameter to reduce the communication cost. This is also the first work that is communication-efficient for adaptive decentralized learning. 
	\item We establish the convergence rate of our proposed two approaches. The extensive experimental results confirm the effectiveness of our proposed approaches.
\end{itemize}

\vspace{-3pt}
\section{Related Works}
\vspace{-3pt}
\textbf{Decentralized Training} The decentralized training approaches have attracted increasing interest in recent years due to its promising performance in communication. Recently, numerous approaches \cite{alghunaim2019linearly,assran2018stochastic,li2019communication,lian2017can,yu2019linear} have been proposed. For instance, \cite{lian2017can} proposes the gossip-based decentralized SGD. Their results show that the leading term of the convergence rate of their approach is as good as that of the centralized SGD approach for non-convex problems. \cite{alghunaim2019linearly} studies the decentralized proximal gradient descent approach for the non-smooth problems and gets a linear convergence rate. \cite{assran2018stochastic} proposes a stochastic gradient push approach to make the decentralized optimization robust to stragglers and communication delays. However, all these approaches only focus on using the same learning rates across different worker nodes. 

\textbf{Adaptive Learning Rate} Although SGD has shown impressive performance for some machine learning tasks, yet it does not perform well for some other tasks, such as the training of graph convolutional neural networks \cite{kipf2016semi} and deep factorization machines \cite{guo2017deepfm}. Especially, \cite{zhang2019adam} shows that SGD is not suitable for the task where the stochastic gradient has the heavy-tailed noise. On the contrary, the adaptive approach, such as AdaGrad \cite{duchi2011adaptive}, Adam \cite{kingma2014adam}, AMSGrad \cite{reddi2019convergence}, can perform well. These approaches compute the learning rate from the historical stochastic gradient automatically so that they can capture the specific properties of the data. Recently, \cite{reddi2020adaptive} applies the adaptive learning rate to centralized training. However, this approach only uses the adaptive learning rate on the server and still adopts the same constant learning rate on the worker nodes. Thus, this approach simplifies the adaptive learning rate too much and cannot fully exploit the benefit of the adaptive learning rate. Furthermore, there are no works applying the adaptive learning rate for decentralized training. To the best of our knowledge, our work is the first adaptive decentralized training approach. 

\textbf{Efficient Communication} For the decentralized training, the communication overhead might be the bottleneck when the model is large. To address this issue, some communication-efficient approaches \cite{koloskova2019decentralizedcon,koloskova2019decentralizednonconv,tang2018communication,li2019communication} have been proposed in recent years.  To reduce the communication overhead, two strategies \cite{alistarh2017qsgd,wen2017terngrad,bernstein2018signsgd,yu2019linear,stich2018sparsified} are commonly used. The first one is to reduce the number of communication rounds. For instance, \cite{li2019communication} proposes a decentralized local SGD approach, which combines the federated averaging and decentralized SGD to reduce the number of communication rounds. The other strategy is to reduce the communication cost in each communication round, such as quantizing the parameter to consume fewer bits. For instance, \cite{koloskova2019decentralizedcon} proposes the communication-efficient decentralized SGD for convex problems, while \cite{tang2018communication,koloskova2019decentralizednonconv} studies the convergence rate of  the communication-efficient decentralized SGD for non-convex problems. However, all these approaches only focus on the regular decentralized SGD, ignoring the more challenging adaptive decentralized training approaches. %Therefore, it is necessary to study 

\vspace{-3pt}
\section{Preliminary Knowledge}
\vspace{-3pt}
\subsection{Problem Setup}
\vspace{-3pt}
In this paper, we focus on the decentralized training approach with the adaptive learning rate to solve Eq.~(\ref{lossfunction}). Specifically, each worker node optimizes the following subproblem: 
\begin{equation} \label{local_lossfunction}
\begin{aligned}
&	\min_{\mathbf{x}} f^{(k)}(\mathbf{x})  \triangleq \mathbb{E}_{\xi\sim\mathcal{D}^{(k)}} F^{(k)}(\mathbf{x}; \xi)  \ ,\\
\end{aligned}
\end{equation}
where $f^{(k)}(\mathbf{x})$ is the loss function on the $k$-th worker node.  Then, these worker nodes collaboratively learn Eq.~(\ref{lossfunction}).
In a decentralized training system,  there is no central server node. All worker nodes only need to communicate with their neighbor nodes. Formally,  the communication protocol is determined by a graph $\mathcal{G}=\{V, W\}$ where $V=[K]$ denotes the set of worker nodes and $W=[w_{ij}]\in \mathbb{R}^{K\times K}$ represents the connection between different worker nodes. Specifically, $w_{ij}>0$ indicates that the $i$-th worker node and the $j$-th worker node are connected so that they can communicate with each other, while $w_{ij}=0$ denotes that these two worker nodes are disconnected so that no communication is needed.  In addition, following \cite{koloskova2019decentralizedcon,koloskova2019decentralizednonconv,li2019communication}, the matrix $W$ satisfies the following properties.
\begin{definition} \label{graph}
	$W$ is symmetric and doubly stochastic, i.e. $W^T=W$, $W\mathbf{1}=\mathbf{1}$, and $\mathbf{1}^TW=\mathbf{1}^T$. In addition, its eigenvalues satisfy that $|\lambda_n|\leq \cdots \leq |\lambda_2|\leq |\lambda_1|=1$ and it has the spectral gap $\rho=1-|\lambda_2| \in (0, 1]$.
\end{definition}

\vspace{-3pt}
\subsection{Notations}
\vspace{-3pt}
In this paper, the vector is represented by the lowercase bold letter and the matrix is represented by the uppercase letter.  
In addition, for a vector $\mathbf{x}_t\in \mathbb{R}^d$, we use  $\mathbf{x}_{t, j}$ or  $[\mathbf{x}_t]_j$ to represent the $j$-th coordinate, and $\sqrt{ \mathbf{x}_t}$ to represent the element-wise square root. For two vectors $\mathbf{x}_t\in \mathbb{R}^d$ and $\mathbf{y}_t\in \mathbb{R}^d$, $\frac{\mathbf{x}_t}{\mathbf{y}_t}$ denotes the element-wise division (assume all elements are non-zero). $\circ$ denotes  the element-wise product.
%Here, we summarize the notations used in this paper. 
Moreover, $\mathbf{x}_t^{(k)}$ is the model parameter held by the $k$-th worker node at the $t$-th iteration.
$\bar{\mathbf{x}}_t = \frac{1}{K}\sum_{k=1}^{K}\mathbf{x}_t^{(k)}$ is the averaged model parameter at the $t$-th iteration.
$\mathbf{g}_t^{(k)} \triangleq \nabla F^{(k)}(\mathbf{x}_t^{(k)}; \xi_t^{(k)} ) $ is the stochastic gradient of the $k$-th worker node with respect to $\mathbf{x}_t^{(k)}$. 
$\nabla f^{(k)}(\mathbf{x}_t^{(k)}) = \mathbb{E}_{\xi\sim\mathcal{D}^{(k)}} \nabla F^{(k)}(\mathbf{x}_t^{(k)}; \xi_t^{(k)}) $ is the full gradient of  the $k$-th worker node with respect to $\mathbf{x}_t^{(k)}$. 
$\nabla f(\mathbf{x}_t)=\frac{1}{K}\sum_{k=1}^{K}\nabla f^{(k)}(\mathbf{x}_t^{(k)}) $ is the full gradient at  the $t$-th iteration.  
$f_*$ represents the minimum value of Eq.~(\ref{lossfunction}).
$p$ denotes the communication period.

\vspace{-3pt}
\section{Adaptive Decentralized Training Approach}

\vspace{-3pt}
\subsection{Decentralized Adam}
\vspace{-3pt}

In Alg.~\ref{decentralized_adam}, we propose the decentralized Adam, which employs the adaptive learning rate for each worker node in a decentralized training system. In detail, as shown in Alg.~\ref{decentralized_adam}, there are two components. In the local computation component (Line 3-6),  like Adam \cite{kingma2014adam}, each worker node computes the stochastic gradient $\mathbf{g}_t^{(k)}$ based on its own data and updates its intermediate  model parameter $\mathbf{x}_{t+\frac{1}{2}}^{(k)}$ as follows:
\begin{equation}
\mathbf{x}_{t+\frac{1}{2}}^{(k)}=\mathbf{x}_{t}^{(k)} - \eta\frac{\mathbf{m}_t^{(k)}}{\sqrt{\mathbf{v}_{t}^{(k)} } + \tau} \ ,
\end{equation}
where $\tau>0$ avoids the zero denominator, $\eta$ is the initial learning rate which is adjusted by $\sqrt{\mathbf{v}_{t}^{(k)}}$. $\mathbf{v}_{t}^{(k)} $ and $\mathbf{m}_t^{(k)}$ are computed  by using the historical gradient  as shown in Line 4-5 of Alg.~\ref{decentralized_adam}.  In the communication component (Line 7-11), the worker nodes conduct communication with neighbor nodes at every $p$ iterations. Specifically, when $\text{mod}(t+1, p)=0$, each worker node updates its model parameter by aggregating the intermediate model parameter of its neighbors as follows:
\begin{equation}
\mathbf{x}_{t+1}^{(k)} =\sum_{j\in \mathcal{N}_k}w_{kj}\mathbf{x}_{t+\frac{1}{2}}^{(j)} \ ,
\end{equation}
where $\mathcal{N}_k$ denotes the neighbor nodes of the $k$-th node.  
When $\text{mod}(t+1, p)\neq0$, the communication is skipped so that  $\mathbf{x}_{t+1}^{(k)} =\mathbf{x}_{t+\frac{1}{2}}^{(k)}$. It can be seen that if $p=1$, the communication will be conducted at each iteration. If $p>1$, some communication rounds are skipped. In other words, each worker node conducts multiple local computation and then performs communication. Hence, a large $p$ indicates the less communication frequency. 

\setlength{\textfloatsep}{2pt}
\begin{algorithm}
	\caption{Decentralized Adam (D-Adam)}
	\label{decentralized_adam}
	\begin{multicols}{2}
		\begin{algorithmic}[1]
			\REQUIRE $\mathbf{x}_{0}^{(k)}$, $1>\tau>0$, $p\geq1$, $\eta>0$, $\beta_1 \in [0, 1]$, $\beta_2\in [0, 1]$, $W$.  (same for all workers)
			\STATE For all workers $k$, do:
			\FOR{$t=0,\cdots, T-1$}
			
			\STATE Compute gradient $\mathbf{g}_t^{(k)}=\nabla F(\mathbf{x}_{t}^{(k)}; \xi_{t}^{(k)})$ 
			\STATE  $\mathbf{m}_{t}^{(k)} = \beta_1\mathbf{m}_{t-1}^{(k)} + (1-\beta_1) \mathbf{g}_t^{(k)}$ 
			\STATE  $\mathbf{v}_{t}^{(k)} = \beta_2\mathbf{v}_{t-1}^{(k)} + (1-\beta_2)\mathbf{g}_t^{(k)}\circ \mathbf{g}_t^{(k)}$ 
			\STATE $\mathbf{x}_{t+\frac{1}{2}}^{(k)}=\mathbf{x}_{t}^{(k)} - \eta\frac{\mathbf{m}_t^{(k)}}{\sqrt{\mathbf{v}_{t}^{(k)} } + \tau}$
			\IF {mod($t+1$, $p$)=0}
			\STATE $\mathbf{x}_{t+1}^{(k)} =\sum_{j\in \mathcal{N}_k}w_{kj}\mathbf{x}_{t+\frac{1}{2}}^{(j)}$
			\ELSE 
			\STATE $\mathbf{x}_{t+1}^{(k)} =\mathbf{x}_{t+\frac{1}{2}}^{(k)}$
			\ENDIF
			\ENDFOR
		\end{algorithmic}
	\end{multicols}
\end{algorithm}

\vspace{-3pt}
\subsection{Decentralized Adam with Compressed Communication}
\vspace{-3pt}

In Alg.~\ref{decentralized_adam}, different workers communicate the full-precision model parameter with each other. When the model size is large, the communication will become the bottleneck. To address this issue, we further propose the decentralized Adam with compressed communication, which compresses the communicated parameter in each communication round. The details are shown in Alg.~\ref{decentralized_adam_com}. In particular, we incorporate the following compression operator to compress the communicated model parameter. 
\begin{definition} \label{compress_op}
	For a compression operator $Q:  \mathbb{R}^d \rightarrow \mathbb{R}^d$, if it satisfies 
	\begin{equation}
	\|\mathbf{x}- Q(\mathbf{x})\|^2\leq (1-\delta) \|\mathbf{x}\|^2 \ ,
	\end{equation}
	where $0<\delta\leq 1$, $Q$ is  $\delta$-contraction. 
\end{definition}
This definition covers the sparsification operator and quantization operator. Both of them are usually used to reduce the communication cost in each communication round. 

\setlength{\textfloatsep}{2pt}
\begin{algorithm}
	\caption{Decentralized Adam with Compressed Communication (CD-Adam)}
	\label{decentralized_adam_com}
	\begin{multicols}{2}
		\begin{algorithmic}[1]
			\REQUIRE $\mathbf{x}_{0}^{(k)}$, $1>\tau>0$, $p\geq1$, $\eta>0$, $\beta_1 \in [0, 1]$, $\beta_2\in [0, 1]$, $W$.  (same for all workers)
			\STATE For all workers $k$, do:
			\FOR{$t=0,\cdots, T-1$}
			\STATE Compute gradient $\mathbf{g}_t^{(k)}=\nabla F(\mathbf{x}_{t}^{(k)}; \xi_{t}^{(k)})$
			\STATE  $\mathbf{m}_{t}^{(k)} = \beta_1\mathbf{m}_{t-1}^{(k)} + (1-\beta_1) \mathbf{g}_t^{(k)}$
			\STATE  $\mathbf{v}_{t}^{(k)} = \beta_2\mathbf{v}_{t-1}^{(k)} + (1-\beta_2)\mathbf{g}_t^{(k)}\circ \mathbf{g}_t^{(k)}$
			\STATE $\mathbf{x}_{t+\frac{1}{2}}^{(k)}=\mathbf{x}_{t}^{(k)} - \eta\frac{\mathbf{g}_t^{(k)}}{\sqrt{\mathbf{v}_{t}^{(k)} } + \tau}$
			\IF {mod($t+1$, $p$)=0}
			\STATE $\mathbf{x}_{t+1}^{(k)} =\mathbf{x}_{t+\frac{1}{2}}^{(k)} + \gamma\sum_{j\in \mathcal{N}_k}w_{kj}(\hat{\mathbf{x}}_{t}^{(j)} - \hat{\mathbf{x}}_{t}^{(k)})$
			\STATE $\mathbf{q}_{t}^{(k)} = Q(\mathbf{x}_{t+1}^{(k)}  - \hat{\mathbf{x}}_{t}^{(k)})$
			\STATE Send $\mathbf{q}_{t}^{(k)} $ and receive $\mathbf{q}_{t}^{(j)}$ for $j\in  \mathcal{N}_k$
			\STATE $\hat{\mathbf{x}}_{t+1}^{(j)} = \hat{\mathbf{x}}_{t}^{(j)} +\mathbf{q}_{t}^{(j)}$ for $j\in  \mathcal{N}_k$
			\ELSE
			\STATE $\mathbf{x}_{t+1}^{(k)} =\mathbf{x}_{t+\frac{1}{2}}^{(k)}$
			\STATE $\hat{\mathbf{x}}_{t+1}^{(j)} = \hat{\mathbf{x}}_{t}^{(j)} $ for $j\in  \mathcal{N}_k$
			\ENDIF
			\ENDFOR
		\end{algorithmic}
	\end{multicols}
\end{algorithm}

In Alg.~\ref{decentralized_adam_com}, there are also two components: the local computation component (Line 3-6) and the communication component (Line 7-15).  The  local computation component  is same with Alg.~\ref{decentralized_adam}, where each worker node computes its own learning rate from the historical gradient and updates  its own intermediate model parameter $\mathbf{x}_{t+\frac{1}{2}}^{(k)}$ which is shown in Line 3-6 of Alg.~\ref{decentralized_adam_com}. For the communication component, inspired by \cite{koloskova2019decentralizedcon}, when $\text{mod}(t+1, p)=0$, each worker node first updates its model parameter  as follows:
\begin{equation}
\mathbf{x}_{t+1}^{(k)} =\mathbf{x}_{t+\frac{1}{2}}^{(k)} + \gamma\sum_{j\in \mathcal{N}_k}w_{kj}(\hat{\mathbf{x}}_{t}^{(j)} - \hat{\mathbf{x}}_{t}^{(k)}) \ ,
\end{equation}
where $\gamma>0$ and $\hat{\mathbf{x}}_{t}^{(k)}$ is the auxiliary model parameter of the $k$-th node for efficient communication. Here, each worker node $k$ stores all $\hat{\mathbf{x}}_{t}^{(j)}$ of its neighbors $j\in \mathcal{N}_k$. Hence, in this step, no communication happens. After obtaining $\mathbf{x}_{t+1}^{(k)}$,  each worker node $k$ communicates $\mathbf{q}_{t}^{(k)} = Q(\mathbf{x}_{t+1}^{(k)}  - \hat{\mathbf{x}}_{t}^{(k)})$ with its neighbors and updates the auxiliary model parameter by $\hat{\mathbf{x}}_{t+1}^{(j)} = \hat{\mathbf{x}}_{t}^{(j)} +\mathbf{q}_{t}^{(j)}$ for all $j\in \mathcal{N}_k$, where $\mathbf{q}_{t}^{(k)}$ is used for controlling the noise introduced by the compression operator \cite{koloskova2019decentralizedcon}. Note that $\mathbf{q}_{t}^{(k)}$ is compressed so that the communication cost is less than that of Alg.~\ref{decentralized_adam}.  Thus, it is more  efficient in communication.  
When $\text{mod}(t+1, p)\neq0$, the communication is skipped. We set  $\mathbf{x}_{t+1}^{(k)} =\mathbf{x}_{t+\frac{1}{2}}^{(k)}$ and $\hat{\mathbf{x}}_{t+1}^{(j)} = \hat{\mathbf{x}}_{t}^{(j)} $ for $j\in  \mathcal{N}_k$ directly. 
%Finally, we establish its convergence rate as follows. 

In summary, we propose two adaptive training approaches for serverless learning. To the best of our knowledge, they are the first ones using adaptive learning rates in the decentralized training regime.

\section{Convergence Analysis}

\vspace{-3pt}
\subsection{Assumptions}
\vspace{-3pt}

\begin{assumption}\label{smooth}
	\textbf{(Smoothness)}: All the loss function on worker nodes are $L$-smooth, i.e.,
	\begin{equation}
	\|\nabla f^{(k)}(\mathbf{x}) - \nabla f^{(k)}(\mathbf{y})\| \leq L\|\mathbf{x}-\mathbf{y}\| , \quad \forall k\in [K], \forall \mathbf{x}\in \mathbb{R}^{d}, \forall \mathbf{y}\in \mathbb{R}^{d} \ .
	\end{equation}
\end{assumption}

\begin{assumption}\label{var1}
	\textbf{(Bounded gradient variance)}: For the loss function on each worker node, there exists $\sigma_{j}>0 \ (\forall j \in [d])$ such that
	\begin{equation}
	\mathbb{E} [|[\nabla F^{(k)}(\mathbf{x}; \xi )]_j - [\nabla f^{(k)}(\mathbf{x})]_j|^2] \leq \sigma_{j}^2, \quad \forall k\in [K], \forall \mathbf{x} \in \mathbb{R}^{d}, \forall j \in [d]\ .
	\end{equation}

\end{assumption}

\begin{assumption} \label{norm}
	\textbf{(Bounded gradient)}: For the loss function on each worker node, we assume there exists $G>0$ such that
	\begin{equation}
	|[ \nabla F^{(k)}(\mathbf{x}; \xi )]_j | \leq G, \quad \forall \mathbf{x} \in \mathbb{R}^{d}, \forall j \in [d], \forall k\in [K]\ .
	\end{equation}
\end{assumption}

\vspace{-3pt}
\subsection{Convergence Rate}
\vspace{-3pt}
\textbf{Challenges} The convergence analysis of Alg.~\ref{decentralized_adam} and Alg.~\ref{decentralized_adam_com} is more difficult than existing decentralized training approaches. In particular, our algorithms employ the adaptive learning rate that is \textit{coupled} with gradients, while existing decentralized approaches use the learning rate that is \textit{decoupled} with gradients. Hence, it is more challenging to study the convergence rate of our algorithms. Additionally, comparing with centralized approaches, the topology coupled with periodic and compressed communication makes the convergence analysis of our algorithms much more difficult. In the following, we present our algorithms' convergence rate and the proof can be found in Appendix~\ref{proof_theorem1} and~\ref{proof_theorem2}. 

\begin{theorem} \label{theorem1}
	For Algorithm~\ref{decentralized_adam}, under Assumption~\ref{smooth}--\ref{norm}, if  choosing  the initial learning rate $\eta<\frac{\tau^2}{3\sqrt{\beta_2}GL}$ and $0<\tau<1$,  we have
	\begin{equation}
	\begin{aligned}
	& \frac{1}{T}\sum_{t=0}^{T-1} \|\nabla  f(\bar{\mathbf{x}}_{t})\|^2 \leq 2(\sqrt{\beta_2}G + 1)\Bigg(\frac{f(\mathbf{x}_0) -   f_*}{\eta T}+  (1+ \frac{4}{\rho^2}) \frac{d\eta^2p^2G^2L^2}{\tau^3} \\
	& \quad \quad \quad \quad \quad \quad + ( G\sqrt{1-\beta_2} + \frac{\eta L}{2})\Big(\frac{3}{\tau^2}\sum_{j=1}^d\sigma_{j}^2+ (1+ \frac{4}{\rho^2}) \frac{6d\eta^2p^2G^2}{\tau^4} \Big)\Bigg)  \ . \\
	\end{aligned}
	\end{equation}
\end{theorem}

\begin{corollary} \label{cor_1}
	For Algorithm~\ref{decentralized_adam}, under Assumption~\ref{smooth}--\ref{norm}, if choosing $\eta=O(\frac{\sqrt{K}}{\sqrt{T}})$,  $\tau=O(\frac{K^{1/2}}{L})$, and $p=O(\frac{T^{1/4}}{K^{c}})$  where $c\geq0$, for sufficiently large $T$, we have
	\begin{equation}
	\frac{1}{T}\sum_{t=0}^{T-1} \|\nabla  f(\bar{\mathbf{x}}_{t})\|^2 \leq O(\frac{1}{\sqrt{KT}}) + O( (1+ \frac{4}{\rho^2}) \frac{1}{K^{2c}\sqrt{KT}}) + O((1+ \frac{4}{\rho^2}) \frac{1}{K^{2c+\frac{1}{2}}T})
	\end{equation}
\end{corollary}

\begin{remark}
	When $c>0$, the first term $O(\frac{1}{\sqrt{KT}}) $ in Corollary~\ref{cor_1} dominates the other terms.  It indicates a linear speedup with respect to the number of workers. The spectral gap $\rho$ of the topology graph only affects the higher-order terms in the convergence rate. This is consistent with the standard decentralized SGD. When $c=0$, the spectral gap $\rho$ affects the leading term of the convergence rate, slowing down the convergence. In other words, a small communication period $p$ (corresponding to a large $c$) can speed up the convergence. 
\end{remark}

\begin{theorem} \label{theorem2}
	For Algorithm~\ref{decentralized_adam_com}, under Assumption~\ref{smooth}--\ref{norm},  if  choosing  the initial learning rate $\eta<\frac{\tau^2}{3\sqrt{\beta_2}GL}$ and $0<\tau<1$, we have 
	
	\begin{equation}
	\begin{aligned}
	& \frac{1}{T}\sum_{t=0}^{T-1} \|\nabla  f(\bar{\mathbf{x}}_{t})\|^2 \leq    (\sqrt{\beta_2}G + 1)\Bigg(\frac{f(\mathbf{x}_0) -   f(\mathbf{x}_*) }{\eta T}+  (1+\frac{13448}{\rho^4\delta^2})\frac{4d\eta^2 p^2G^2L^2}{\tau^3}\\
	& \quad \quad \quad \quad \quad \quad  + ( G\sqrt{1-\beta_2} + \frac{\eta L}{2})\Big(\frac{3}{\tau^2}\sum_{j=1}^d\sigma_{j}^2 + (1+ \frac{13448}{\rho^4\delta^2}) \frac{24d\eta^2p^2G^2}{\tau^4}  \Big) \Bigg) \ . \\
	\end{aligned}
	\end{equation}
\end{theorem}

\begin{remark}
	Compared with Theorem~\ref{theorem1}, it can be seen that the convergence rate of Alg.~\ref{decentralized_adam_com} has a worse dependence on the spectral gap $\rho$. Moreover, the compression introduces the dependence on the compression coefficient $\delta$.
\end{remark}

\begin{corollary} \label{cor_2}
	For Algorithm~\ref{decentralized_adam_com}, under Assumptions~\ref{smooth}--\ref{norm}, if choosing $\eta=O(\frac{\sqrt{K}}{\sqrt{T}})$,  $\tau=O(\frac{K^{1/2}}{L})$, and $p=O(\frac{T^{1/4}}{K^{c}})$  where $c\geq0$, for sufficiently large $T$, we have
	\begin{equation}
	\frac{1}{T}\sum_{t=0}^{T-1} \|\nabla  f(\bar{\mathbf{x}}_{t})\|^2 \leq O(\frac{1}{\sqrt{KT}}) + O( (1+ \frac{13448}{\rho^4\delta^2}) \frac{1}{K^{2c}\sqrt{KT}}) + O((1+ \frac{13448}{\rho^4\delta^2}) \frac{1}{K^{2c+\frac{1}{2}}T}) \ .
	\end{equation}
\end{corollary}

\begin{remark}
	Similar with Alg.~\ref{decentralized_adam}, when $c>0$, Alg.~\ref{decentralized_adam_com} also obtains the linear speedup  with respect to the number of workers since 
	the first term $O(\frac{1}{\sqrt{KT}}) $ dominates the other terms.  Meanwhile, the spectral gap $\rho$ and  the compression coefficient  $\delta $ only affects the  higher order terms in the convergence rate.  When $c=0$, the convergence rate is slowed down by the spectral gap $\rho$ and  the compression coefficient  $\delta $. 
\end{remark}

\vspace{-3pt}
\section{Experiments}
\vspace{-3pt}

\subsection{Experimental Settings}
\vspace{-5pt}
In our experiments, we use three large-scale datasets, which are described as follows. 

\textbf{CIFAR-10} is an image classification dataset. It has 50,000 images for training set and 10,000 images for testing set. Here, we use this dataset to train ResNet20 \cite{he2016deep}. In particular, the number of epochs is 300. The initial learning rate $\eta$ is set to 0.001, and it is divided by 10 at epoch 150 and  225.
The weight decay is set to $10^{-4}$. The total batch size is set to 128.
\textbf{Criteo\footnote{    \urlstyle{same}\url{http://labs.criteo.com/2014/02/kaggle-display-advertising-challenge-dataset/}} }  is a benchmark advertising dataset for click-through rate (CTR) prediction. The task is to predict whether an advertisement will be clicked or not. There are about 45 million samples. Each sample has 2,086,936 features. Since most features in the advertisement data are categorical, factorization machine \cite{rendle2010factorization}, which learns an embedding vector for each feature, is usually used for CTR prediction. Therefore, the model size is usually large due to the embedding vectors for high-dimensional features and weight matrices.  Here, we use deep factorization machine (DeepFM)  \cite{guo2017deepfm}. In particular,  the embedding dimension of each feature is 10, and there are three fully connected layers (400-400-400). The dropout ratio is set to 0.5, and the initial learning rate $\eta$ is set to 0.001. The total batch size is set to 4096. The number of epochs is 20.
\textbf{Movielens-20M\footnote{\urlstyle{same}\url{https://grouplens.org/datasets/movielens/}}  } is movie recommendation dataset. It  consists of users' ratings for movies. There are about 20 million rating records for 27,000 movies given by 138,000 users. The task is to predict whether a user will give a positive rating for a movie or not. Here, we use the state-of-the-art method Wide\&Deep \cite{cheng2016wide} for this task. In detail, the feature is the concatenation of user's id and movie's id and  the target is the rating. Hence, the feature is also categorical and high-dimensional so that the model size is large. Wide\&Deep sets the embedding dimension of each feature to 10 and the dimension of fully connected layers to 400-400-400. We set the dropout ratio to 0.5, initial learning rate $\eta$ to 0.001, the total batch size to 4096, and the number of  epochs to  20.

In our experiments, we set $\beta_1=0.9$ and $\beta_2=0.999$ for all datasets. In addition, we randomly split samples by 9:1 for training and testing set for Criteo and Movielens-20M, while we use the default setting for CIFAR-10. To evaluate the performance, we use ACC (classification accuracy) for CIFAR-10 and AUC (Area Under the ROC curve) for Criteo and Movielens-20M.  At last, 
all experiments are done with  8 workers (GPUs), which are connected in a ring topology. Thus, the batch size for each worker is 1/8 of the total batch size.

\vspace{-10pt}
\begin{figure}[!htbp]
	\centering 
	\subfigure[CIFAR10]{
		\includegraphics[width=0.258\textwidth]{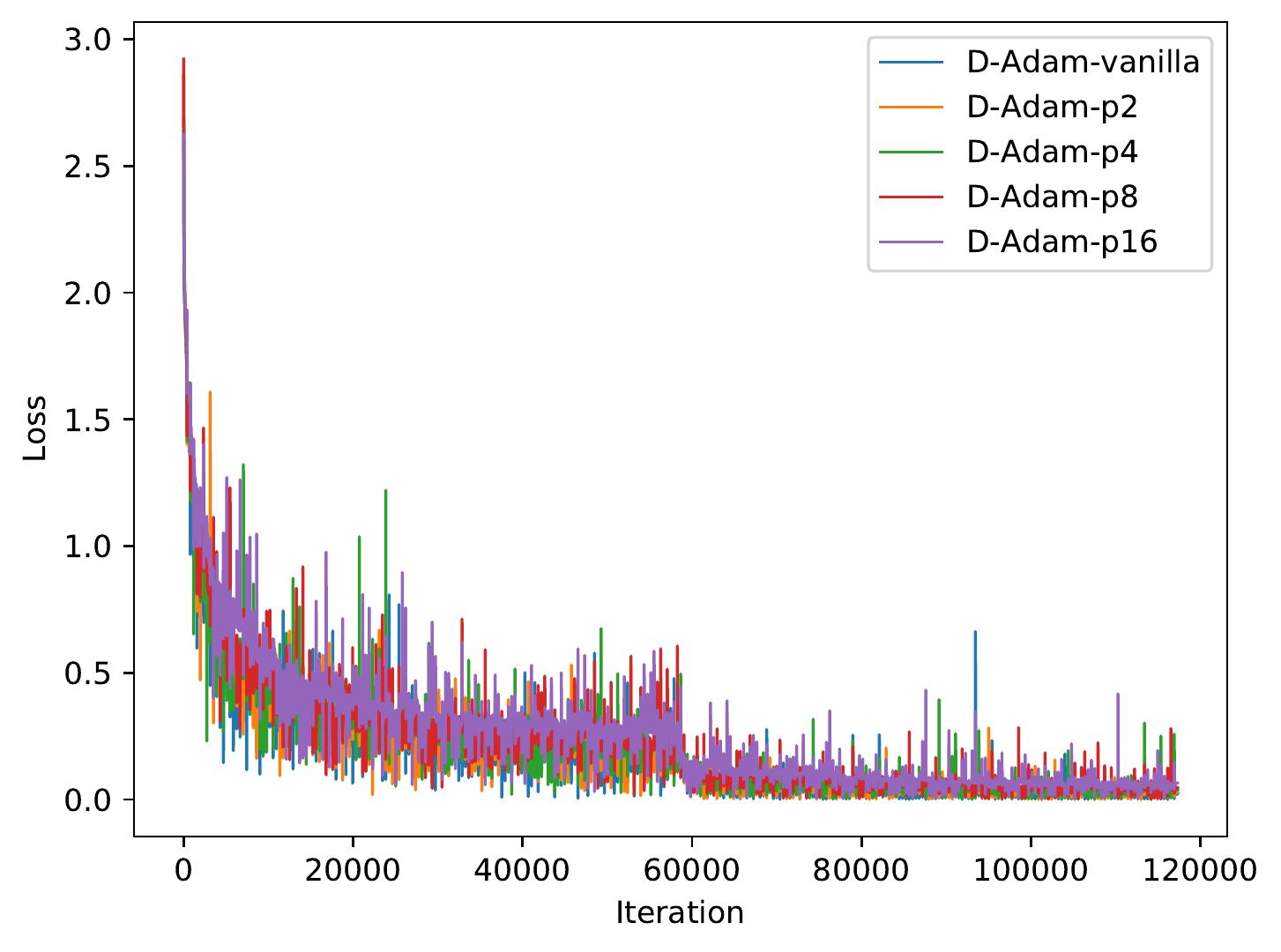}
		\label{train_loss_adam_cifar10}
	}
	\qquad
	\subfigure[Criteo]{
		\includegraphics[width=0.258\textwidth]{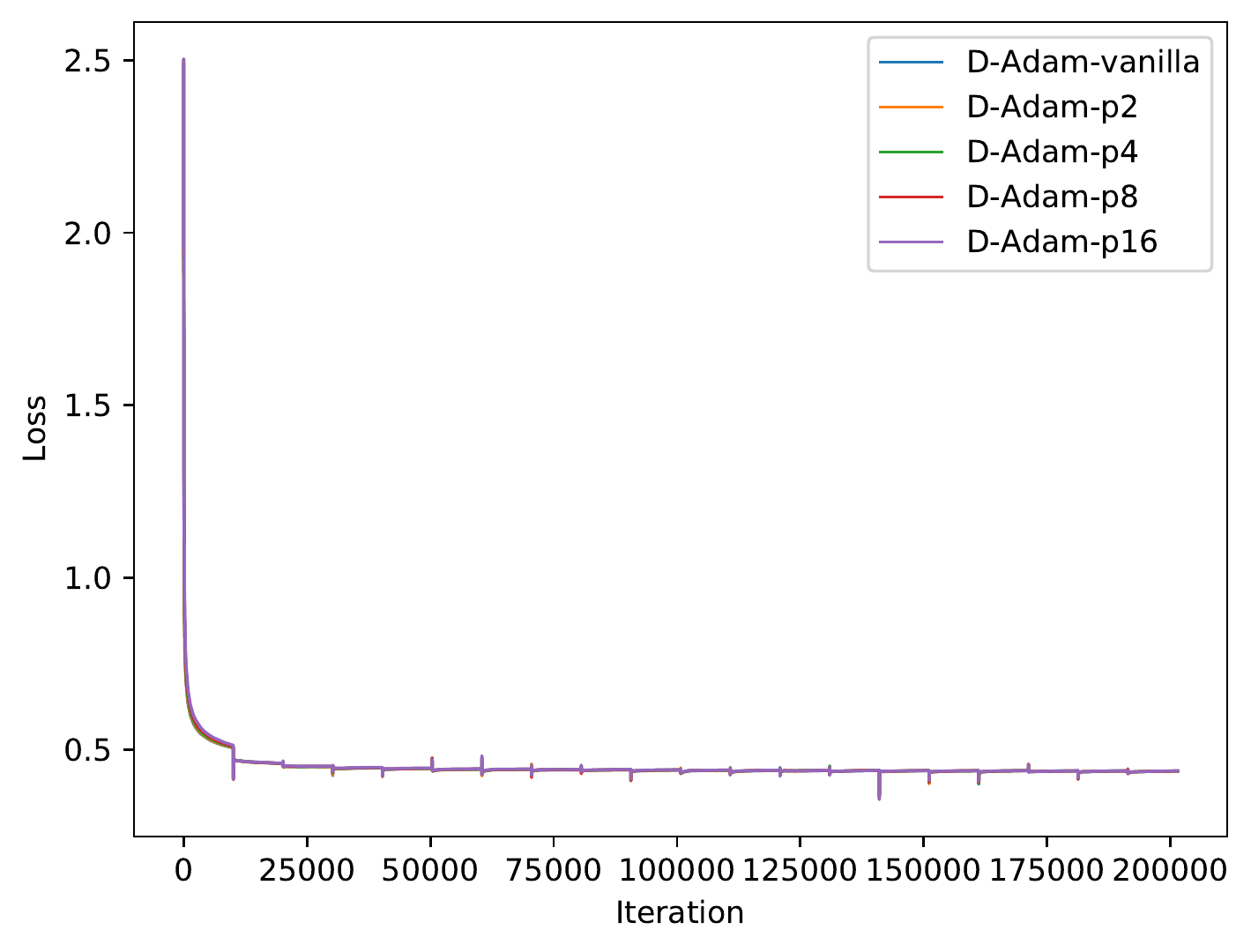}
		\label{train_loss_adam_criteo}
	}
	\qquad
	\subfigure[Movielens-20M]{
		\includegraphics[width=0.258\textwidth]{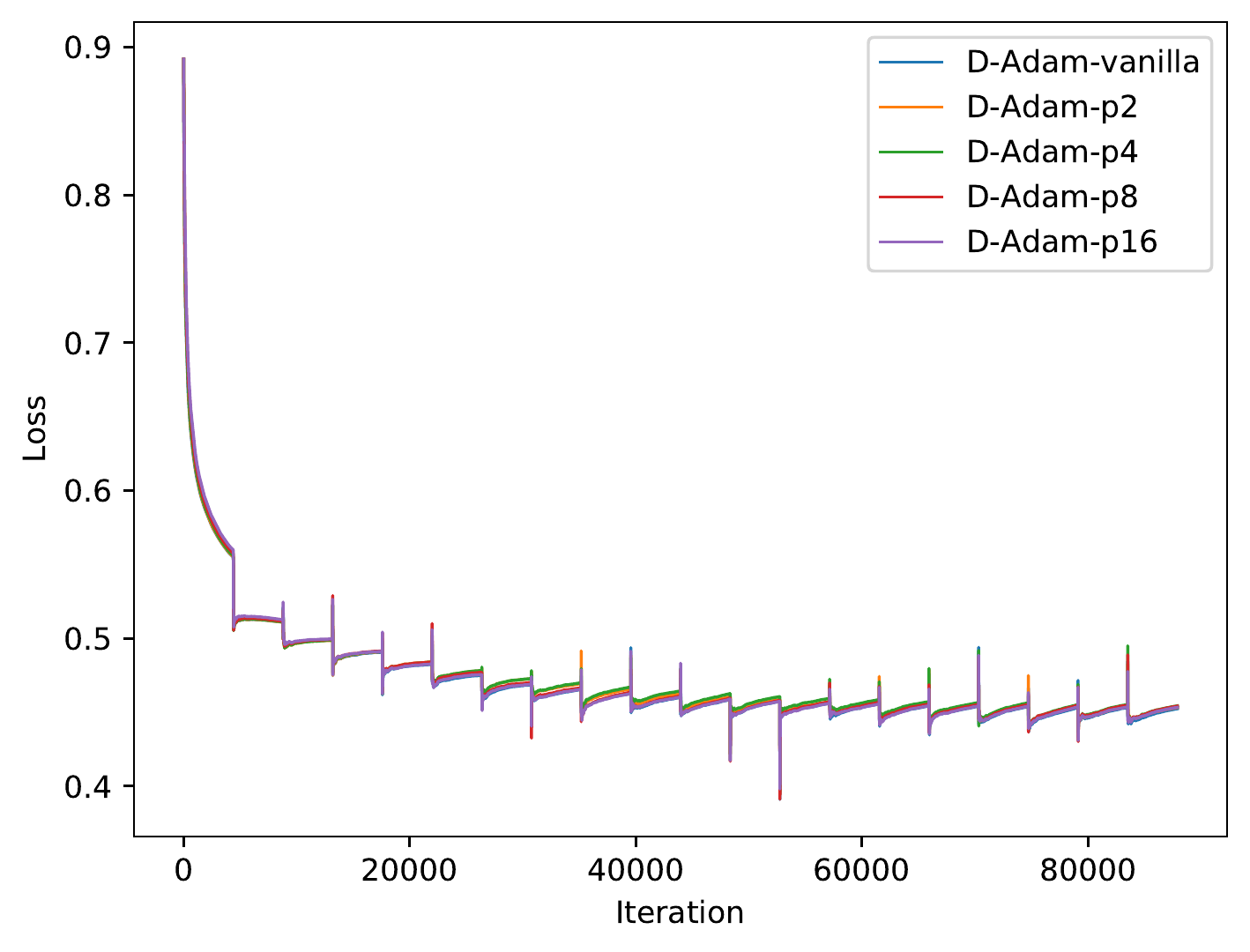}
		\label{train_loss_adam_avazu}
	}
	\vspace{-3pt}
	\caption{ The training performance of D-Adam w.r.t. iterations. }
	\label{convergence_curve_train_adam}
\end{figure}

\vspace{-15pt}
\begin{figure}[!htbp]
	% \vspace{-3pt}
	\centering 
	\subfigure[CIFAR10]{
		\includegraphics[width=0.258\textwidth]{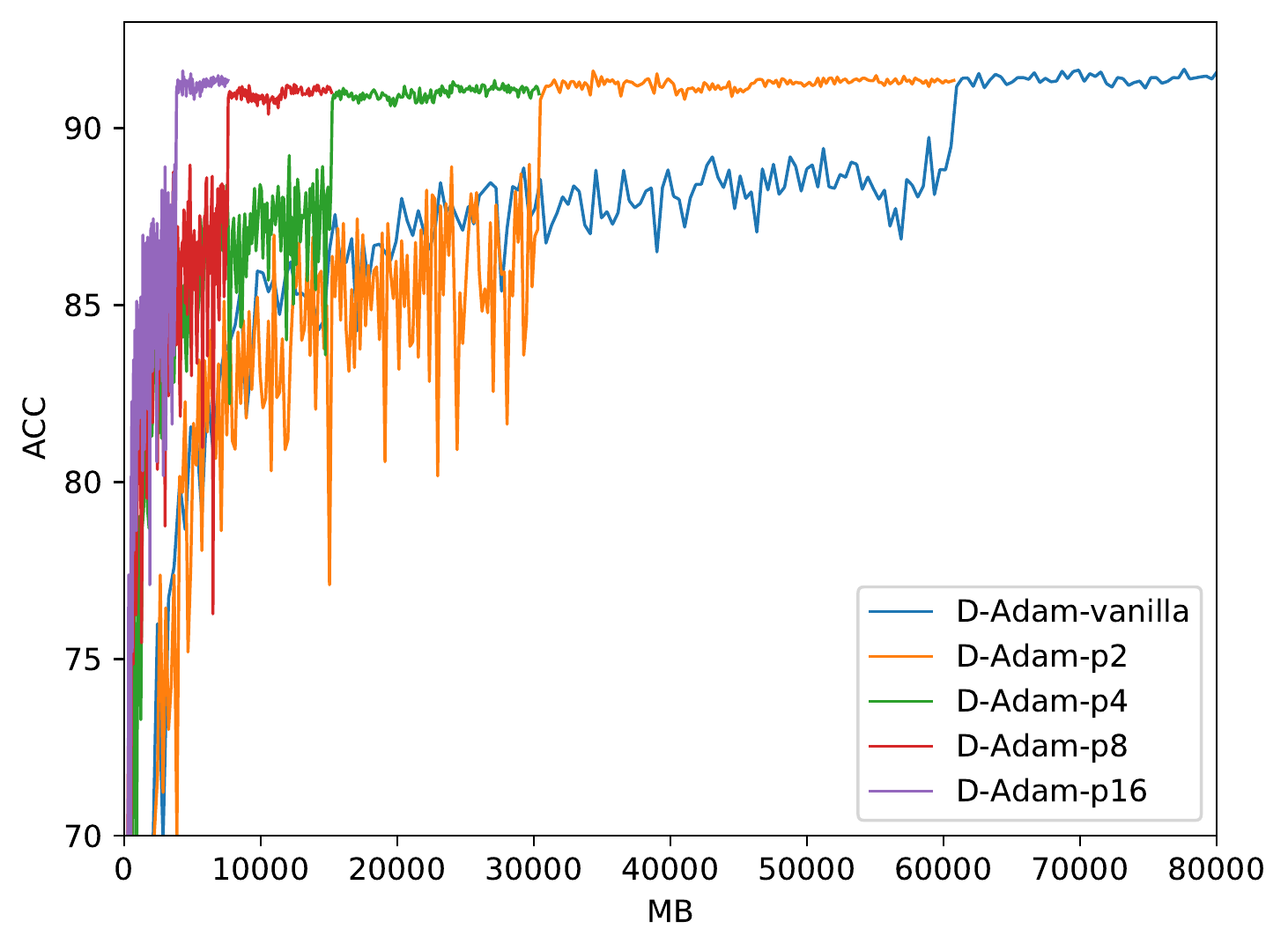}
		\label{test_acc_adam_cifar10}
	}
	\qquad
	\subfigure[Criteo]{
		\includegraphics[width=0.258\textwidth]{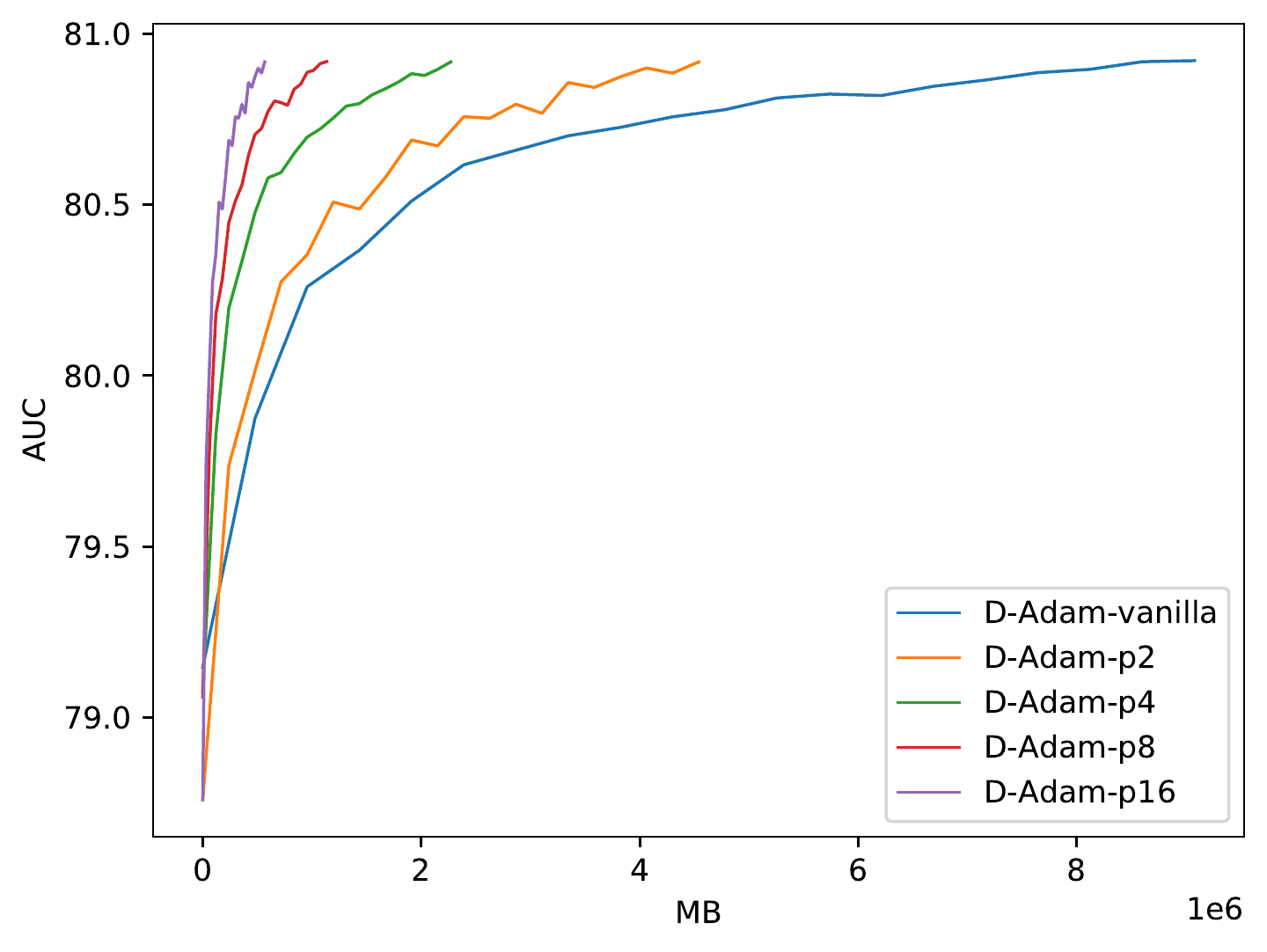}
		\label{test_loss_adam_criteo}
	}
	\qquad
	\subfigure[Movielens-20M]{
		\includegraphics[width=0.258\textwidth]{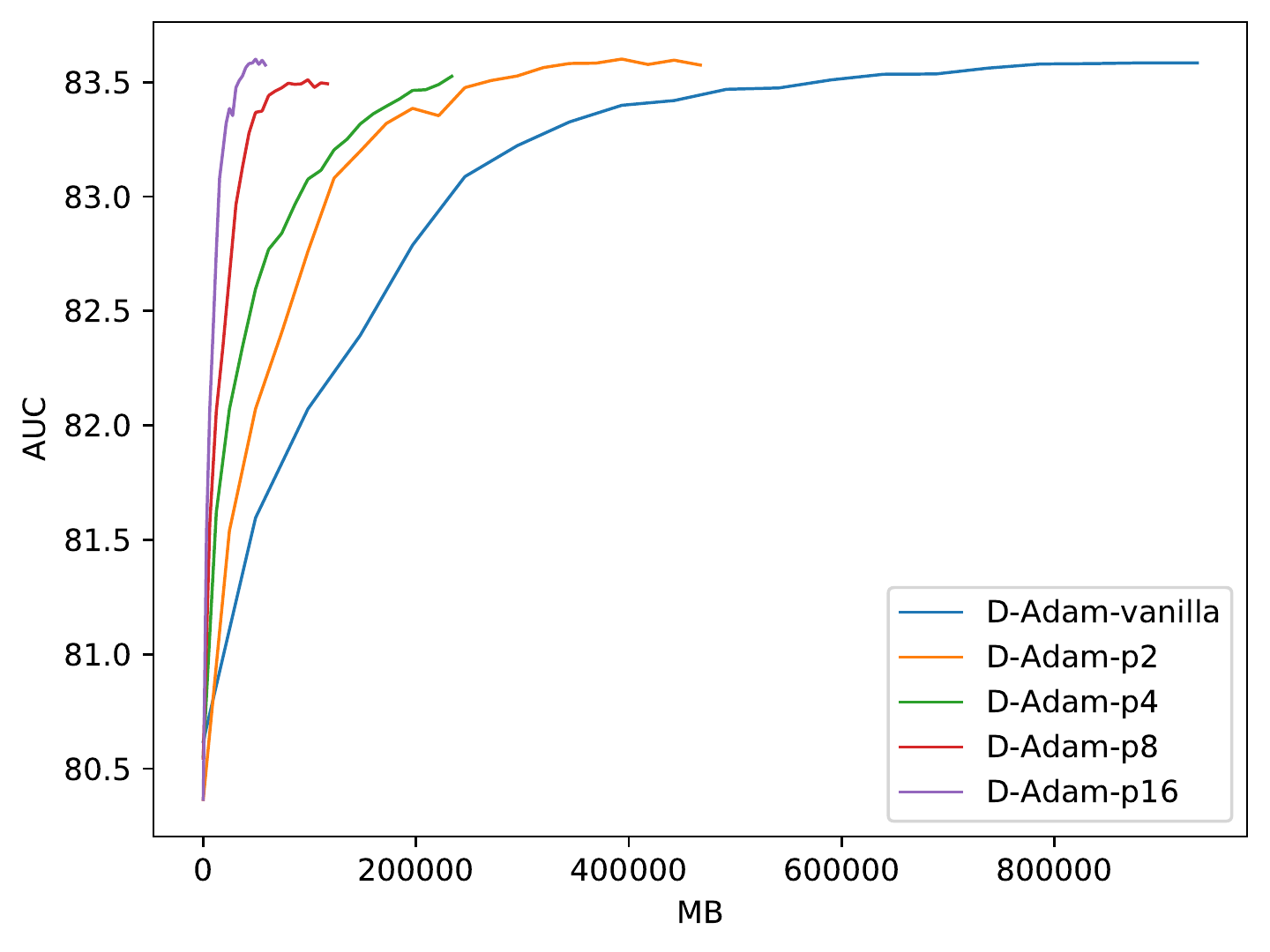}
		\label{test_loss_adam_avazu}
	}
	\vspace{-3pt}
	\caption{ The testing performance of D-Adam w.r.t. communication cost (MB). }
	\label{convergence_curve_test_adam}
\end{figure}

\vspace{-10pt}
\subsection{Experimental Results of Algorithm~\ref{decentralized_adam}}
\vspace{-5pt}
To verify the performance of our proposed D-Adam in Alg.~\ref{decentralized_adam}, we compare it with the vanilla decentralized Adam (D-Adam-vanilla). In particular, like decentralized parallel SGD (D-PSGD) \cite{lian2017can}, D-Adam-vanilla runs Adam in a decentralized  manner and conducts communication at each iteration. For our proposed D-Adam, we  use different communication periods: $p=2,4,8, 16$, and keep the other settings same as D-Adam-vanilla.
In Figure~\ref{convergence_curve_train_adam}, we plot the training loss with respect to the number of iterations, and the testing performance with respect to the communication cost  (MB) is shown in Figure~\ref{convergence_curve_test_adam}.  From these two figures, we have two observations. As for the training performance, our D-Adam algorithms with different communication periods converge to a very  similar  value with that of D-Adam-vanilla. As for the  testing performance, our D-Adam algorithms with different $p$ have almost the same ACC (or AUC) as D-Adam-vanilla. These observations confirm the correctness of our D-Adam. 
Moreover, the larger the communication period $p$ is, the less communication cost  D-Adam needs, which further confirms the communication efficiency of our Alg.~\ref{decentralized_adam}.

\vspace{-5pt}
\subsection{Experimental Results of Algorithm~\ref{decentralized_adam_com}}
\vspace{-5pt}
In this experiment, we conduct experiments to verify the performance of our proposed CD-Adam in Alg.~\ref{decentralized_adam_com}.  Here, 
the compression operator used in our  experiments is the sign operator \cite{bernstein2018signsgd}. The parameter $\gamma$ is set to 0.4.
In Figure~\ref{convergence_curve_train_choco_adam},  we compare the training performance of CD-Adam with that of D-Adam-vanilla. It can be seen that CD-Adam with different communication periods $p$ can converge to almost the same value as the full-precision D-Adam-vanilla, confirming the  correctness of our CD-Adam algorithm.  In Figure~\ref{convergence_curve_test_choco_adam}, we report the convergence result on the testing set w.r.t. the communication cost. Here, we compare CD-Adam with D-Adam whose communication period $p$ is 16 because it has less communication cost than other variants. It can be seen that, with the compressed communication technique, CD-Adam has less communication cost than D-Adam but has almost the same final testing performance with the full-precision D-Adam. In particular, when $p=16$, the difference between the communication cost of CD-Adam and that of D-Adam is very significant. Thus, we can conclude that our Alg.~\ref{decentralized_adam_com} is communication efficient and does not degrade the final testing performance, even though it employs both skipped communication and compressed communication. {More experimental  results can be found in Appendix~\ref{additional_exp}.}

\vspace{-10pt}
\begin{figure}[!htbp]
	\centering 
	\subfigure[CIFAR10]{
		\includegraphics[width=0.258\textwidth]{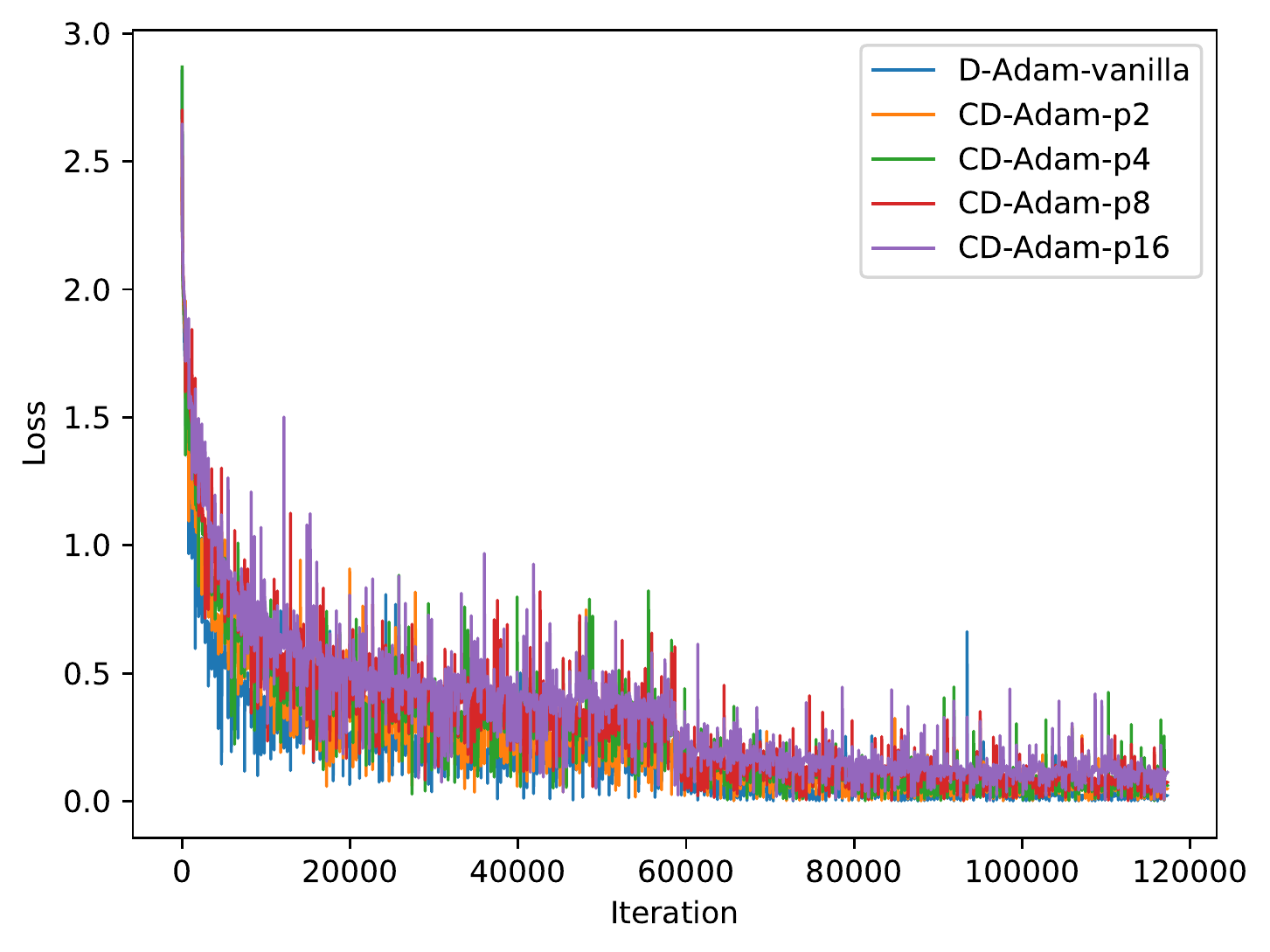}
		\label{train_loss_choco_adam_cifar10}
	}
	\qquad
	\subfigure[Criteo]{
		\includegraphics[width=0.258\textwidth]{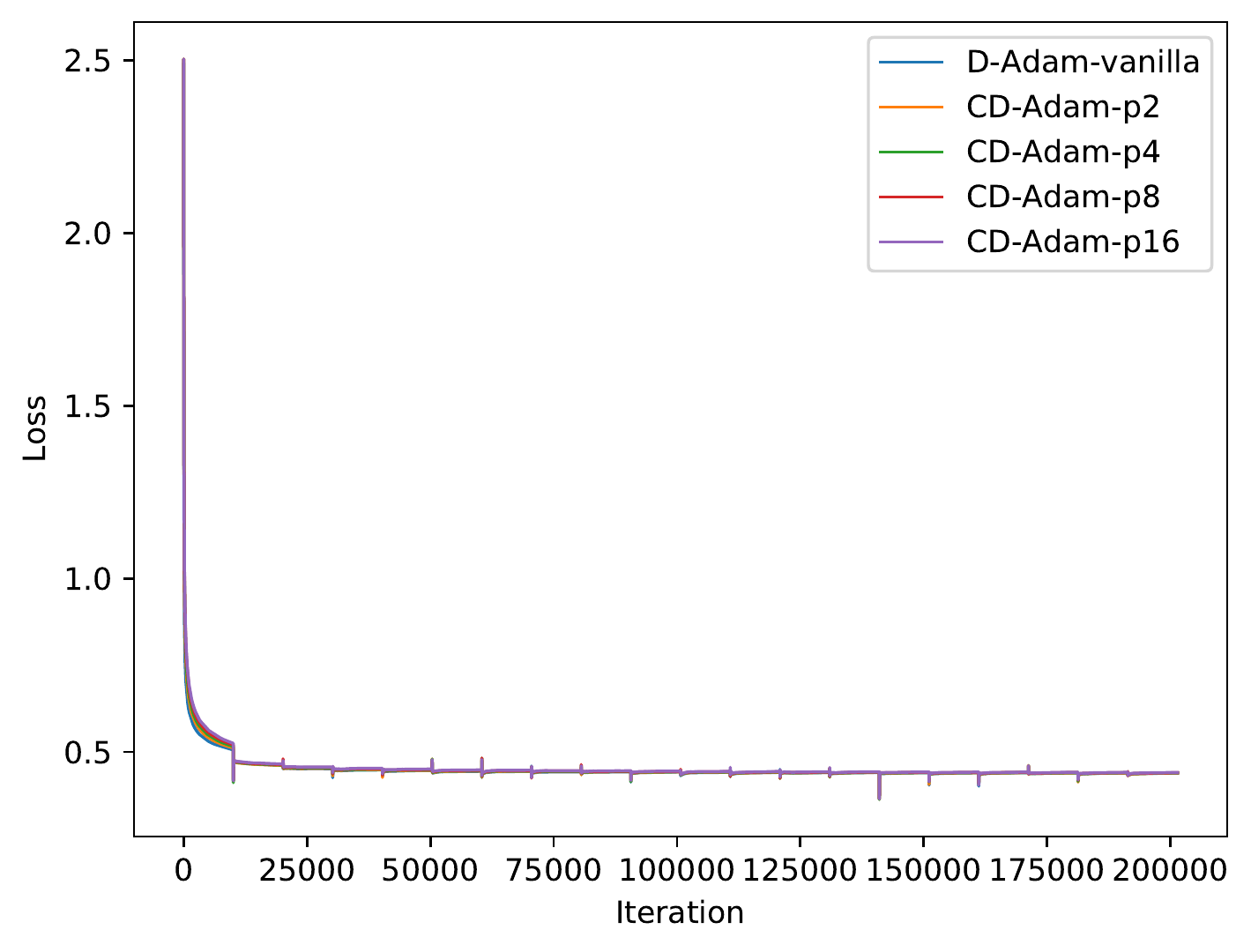}
		\label{train_loss_choco_adam_criteo}
	}
	\qquad
	\subfigure[Movielens-20M]{
		\includegraphics[width=0.258\textwidth]{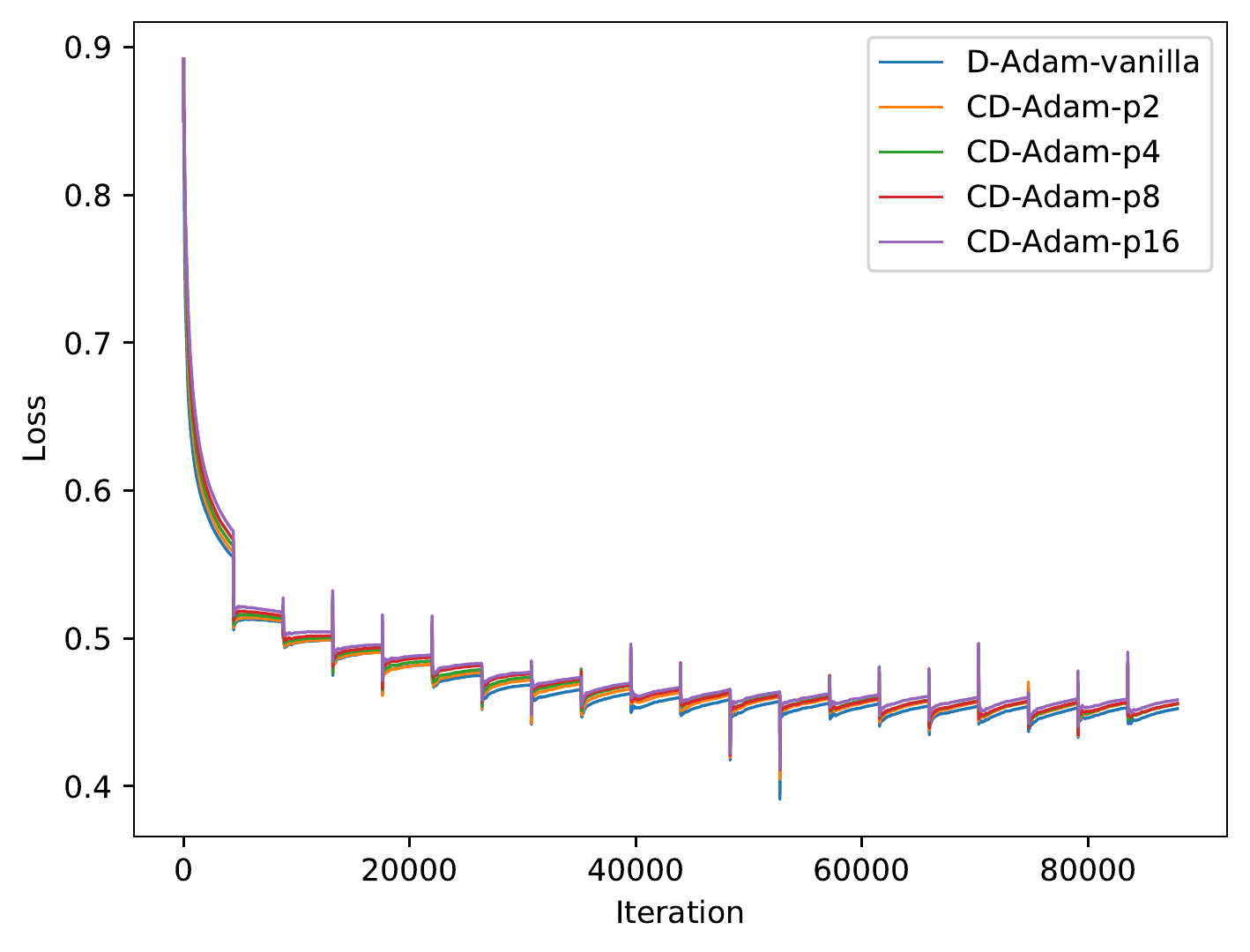}
		\label{train_loss_choco_adam_avazu}
	}
	\vspace{-3pt}
	\caption{  The training performance of CD-Adam w.r.t. iterations. }
	\label{convergence_curve_train_choco_adam}
\end{figure}

\vspace{-10pt}
\begin{figure}[!htbp]
	\vspace{-3pt}
	\centering 
	\subfigure[CIFAR10]{
		\includegraphics[width=0.258\textwidth]{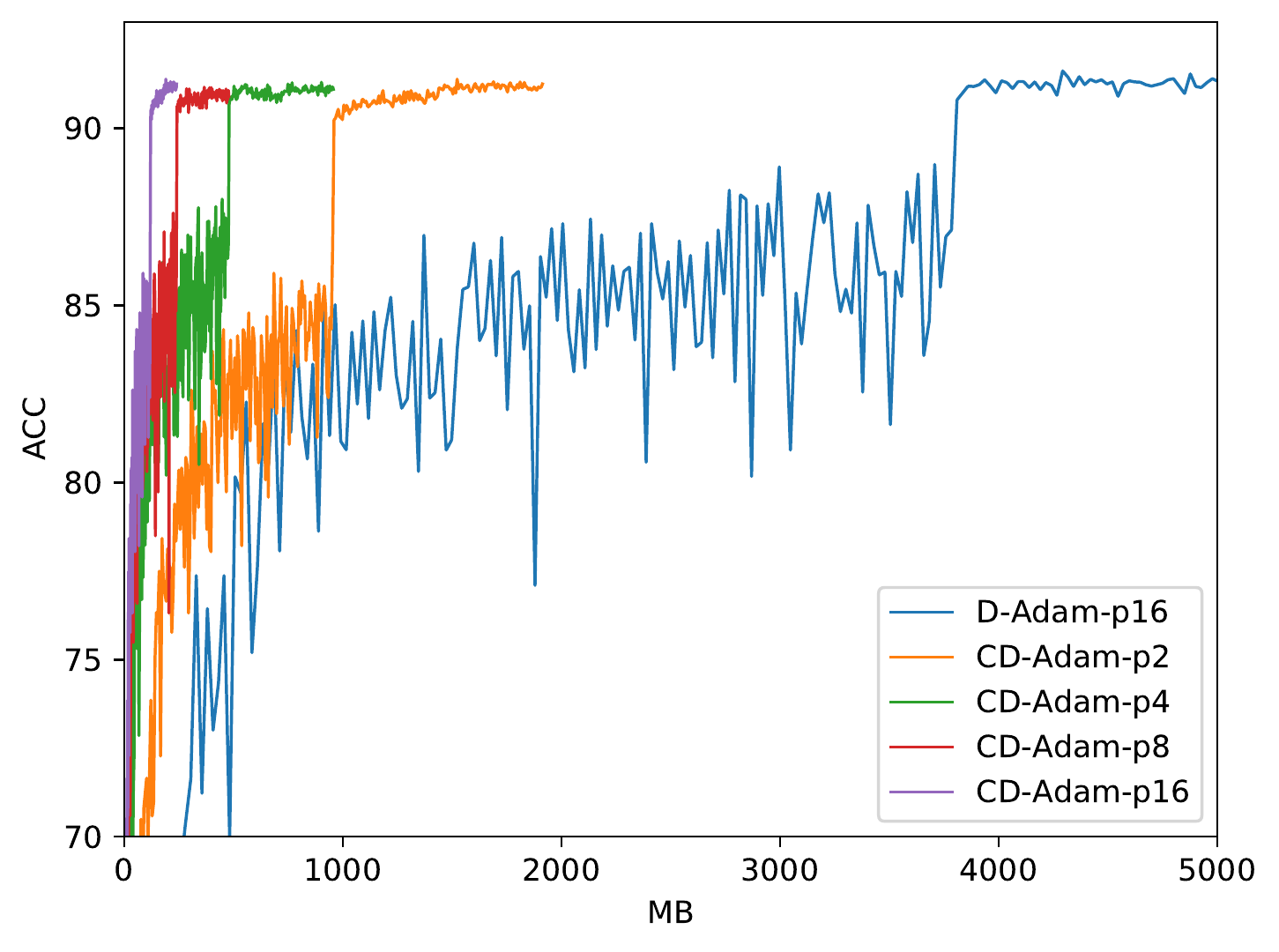}
		\label{test_acc_choco_adam_cifar10}
	}
	\qquad
	\subfigure[Criteo]{
		\includegraphics[width=0.259\textwidth]{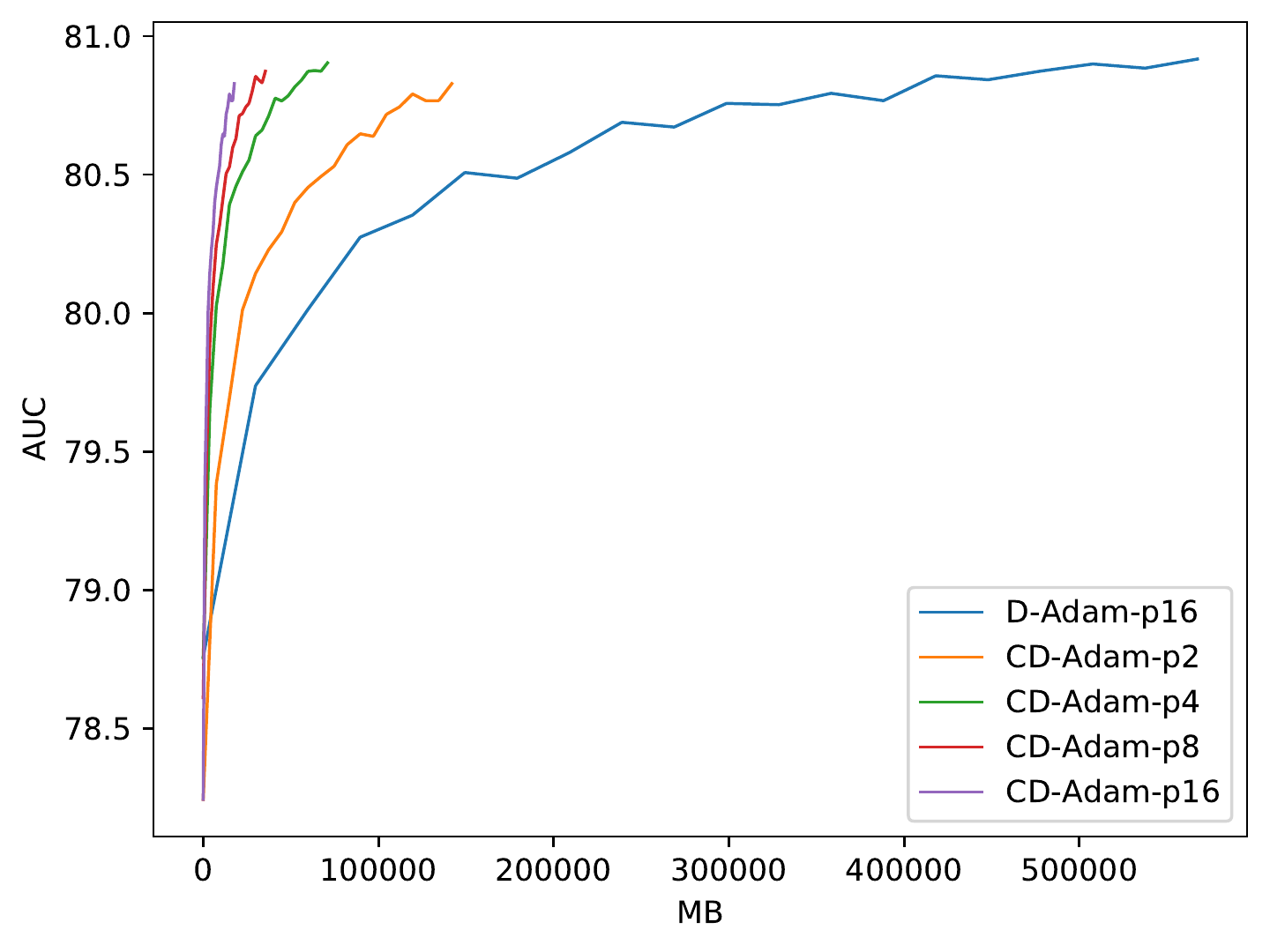}
		\label{test_loss_choco_adam_criteo}
	}
	\qquad
	\subfigure[Movielens-20M]{
		\includegraphics[width=0.258\textwidth]{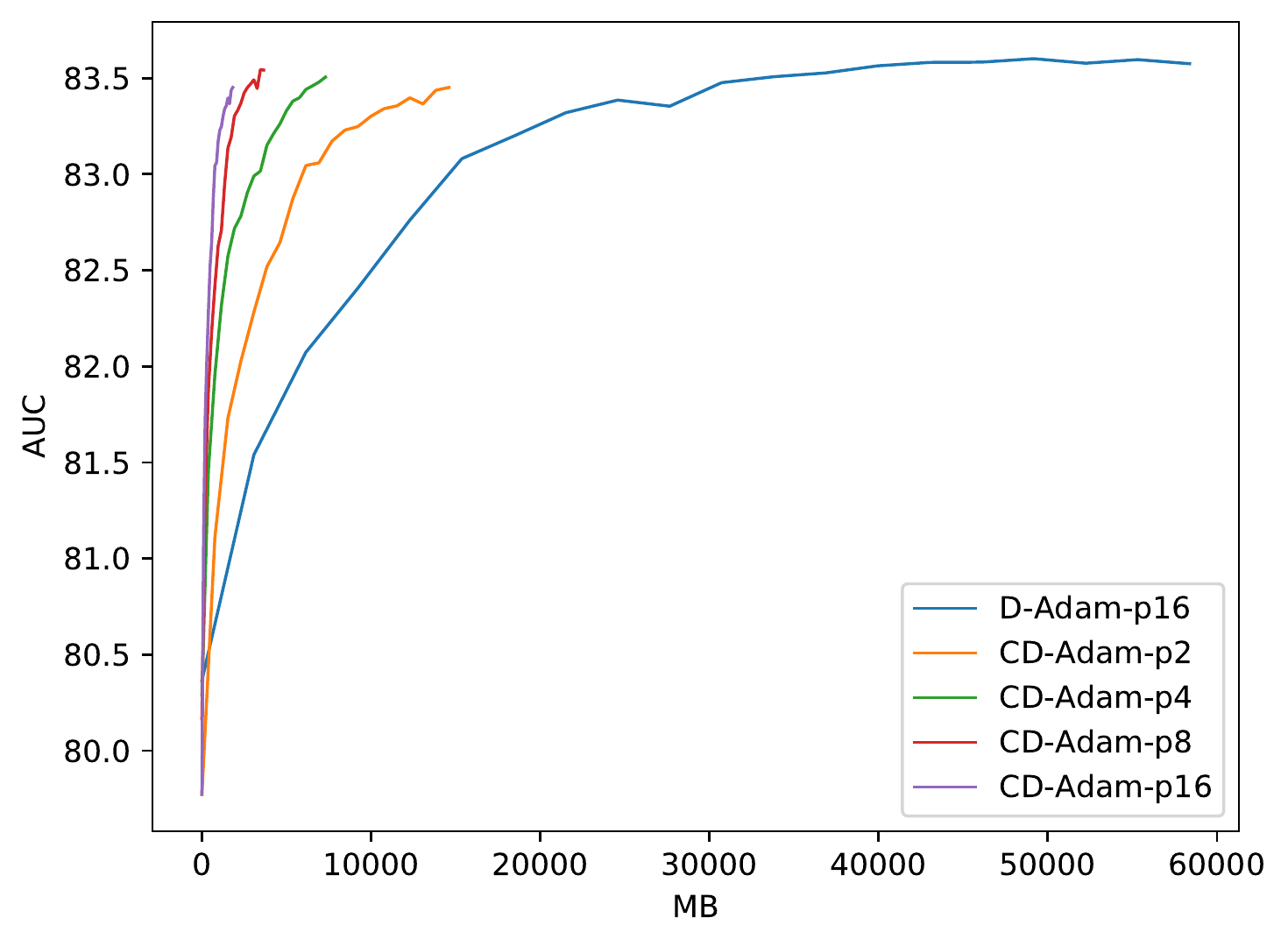}
		\label{test_loss_choco_adam_avazu}
	}
	\vspace{-3pt}
	\caption{ The testing performance of CD-Adam w.r.t. communication cost (MB). }
	\label{convergence_curve_test_choco_adam}
\end{figure}

\vspace{-10pt}
\section{Conclusion}
\vspace{-6pt}
In this paper, we proposed a  novel adaptive decentralized training approach. This is the first research work applying adaptive learning rates to decentralized training.  Moreover, we proposed a communication-efficient adaptive decentralized training algorithm to reduce the number of communication overhead. We also established the convergence rate of our  proposed two algorithms and disclosed when they can achieve linear speedup w.r.t. the number of workers. This is also the first research work obtaining these results. We further validated the proposed methods via various prediction tasks (image classification, click prediction, and recommendation prediction) on three benchmark datasets. All experimental results have verified the effectiveness of our proposed two algorithms.

\bibliographystyle{plain} 
\bibliography{bibfile}

\onecolumn
\input{sec_supp_exp}

\input{sec_supp_proof}

\end{document}

%% file: sec_supp_exp.tex
\section{Appendix}
\subsection{Additional  Experimental Results} \label{additional_exp}
In this subsection, we further show the testing performance of D-Adam and CD-Adam with respect to the number of epochs.  In particular, we compare the  testing performance of  our D-Adam with that of  D-Adam-vanilla in Figure~\ref{convergence_curve_test_adam_epoch}, and also compare  the  testing performance of  our CD-Adam with that of  D-Adam-vanilla in Figure~\ref{convergence_curve_test_choco_adam_epoch}.  From  these two figures, it can  be seen  that our proposed  D-Adam and CD-Adam have very similar final testing performance  with D-Adam-vanilla. In other words, the skipped communication and compressed communication do not hurt the testing performance of our algorithms. 

	\begin{figure}[!htbp]
	% \vspace{-3pt}
	\centering 
	\subfigure[CIFAR10]{
		\includegraphics[width=0.258\textwidth]{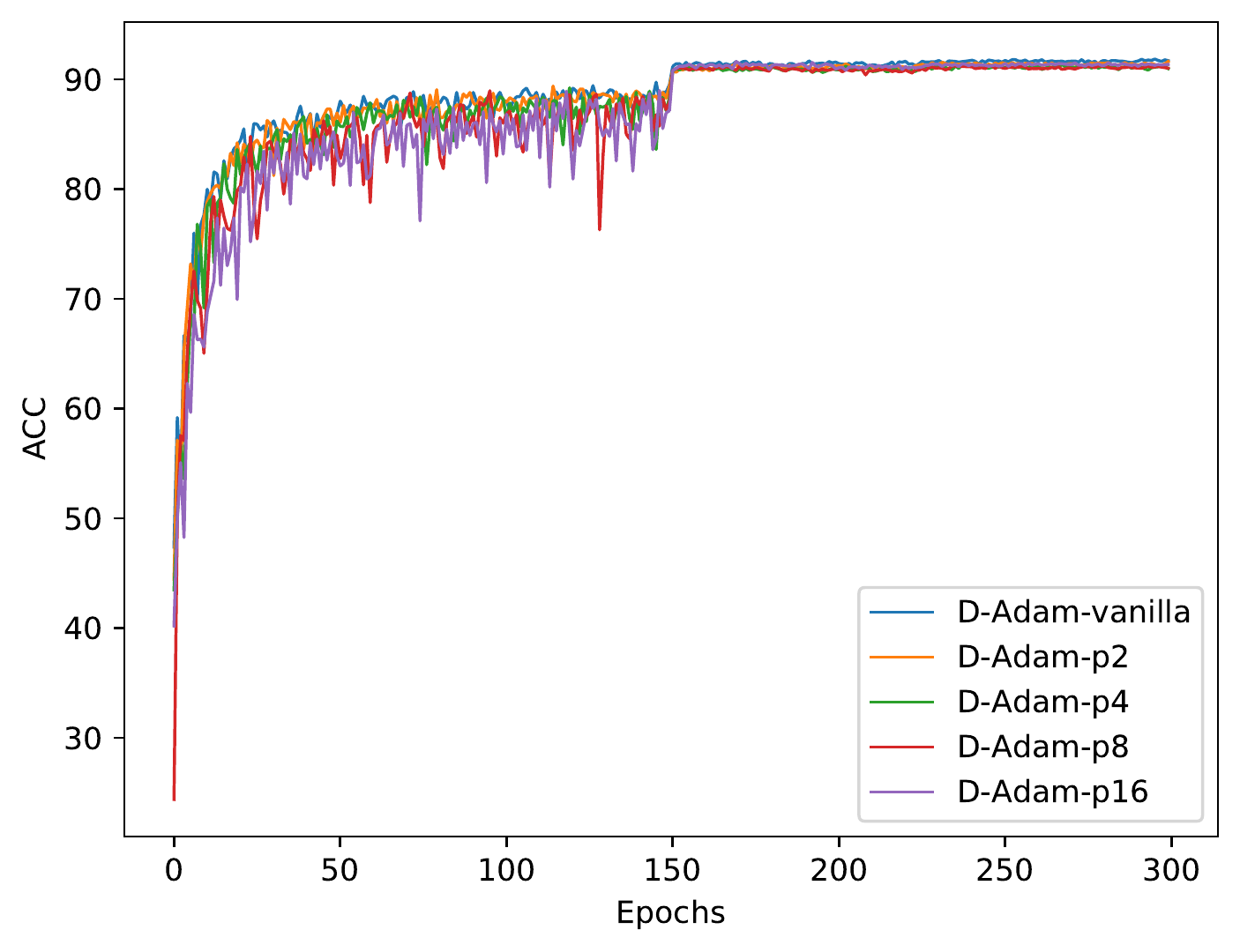}
		\label{test_acc_adam_cifar10_epoch}
	}
	\qquad
	\subfigure[Criteo]{
		\includegraphics[width=0.258\textwidth]{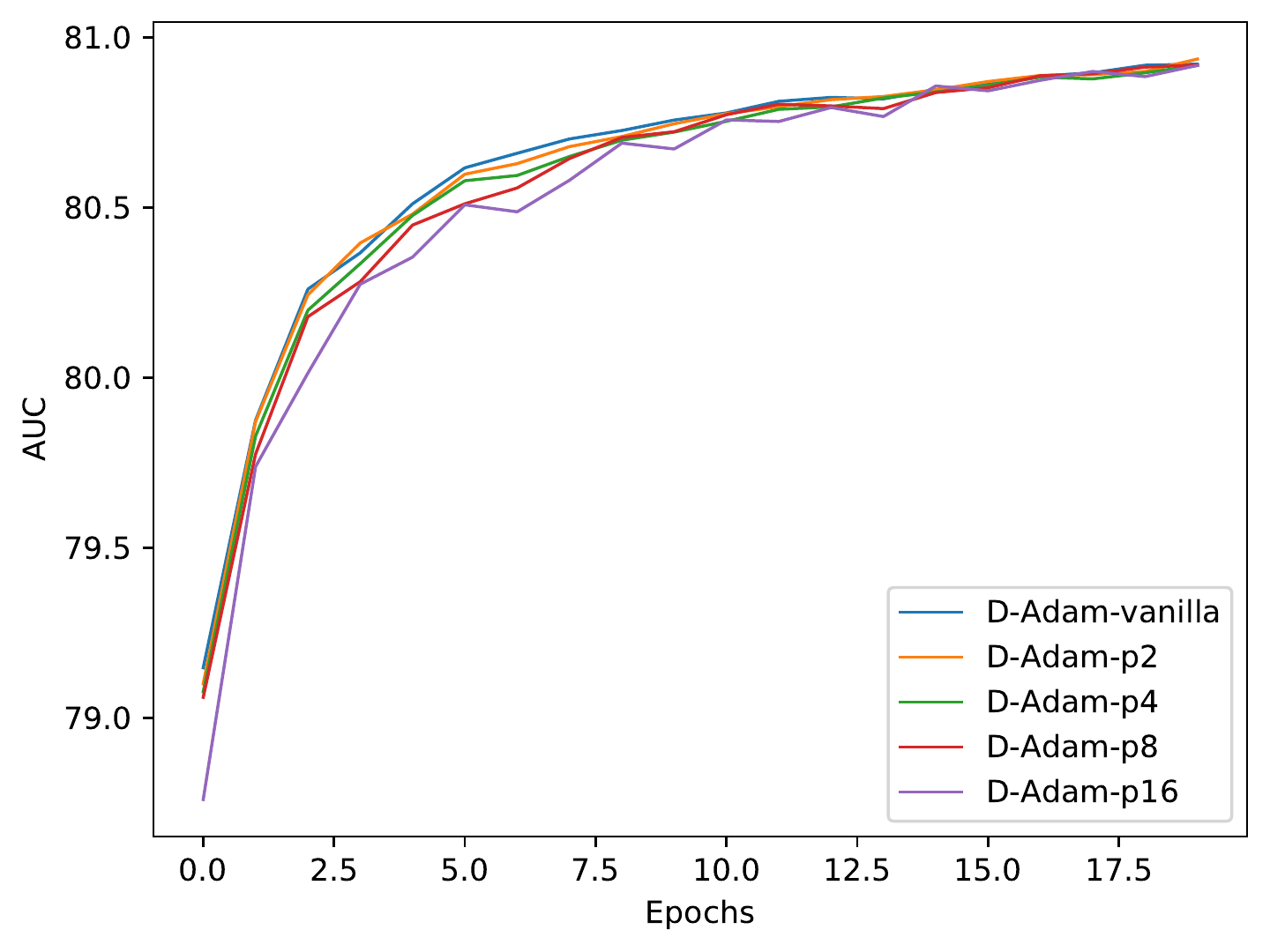}
		\label{test_loss_adam_criteo_epoch}
	}
	\qquad
	\subfigure[Movielens-20M]{
		\includegraphics[width=0.258\textwidth]{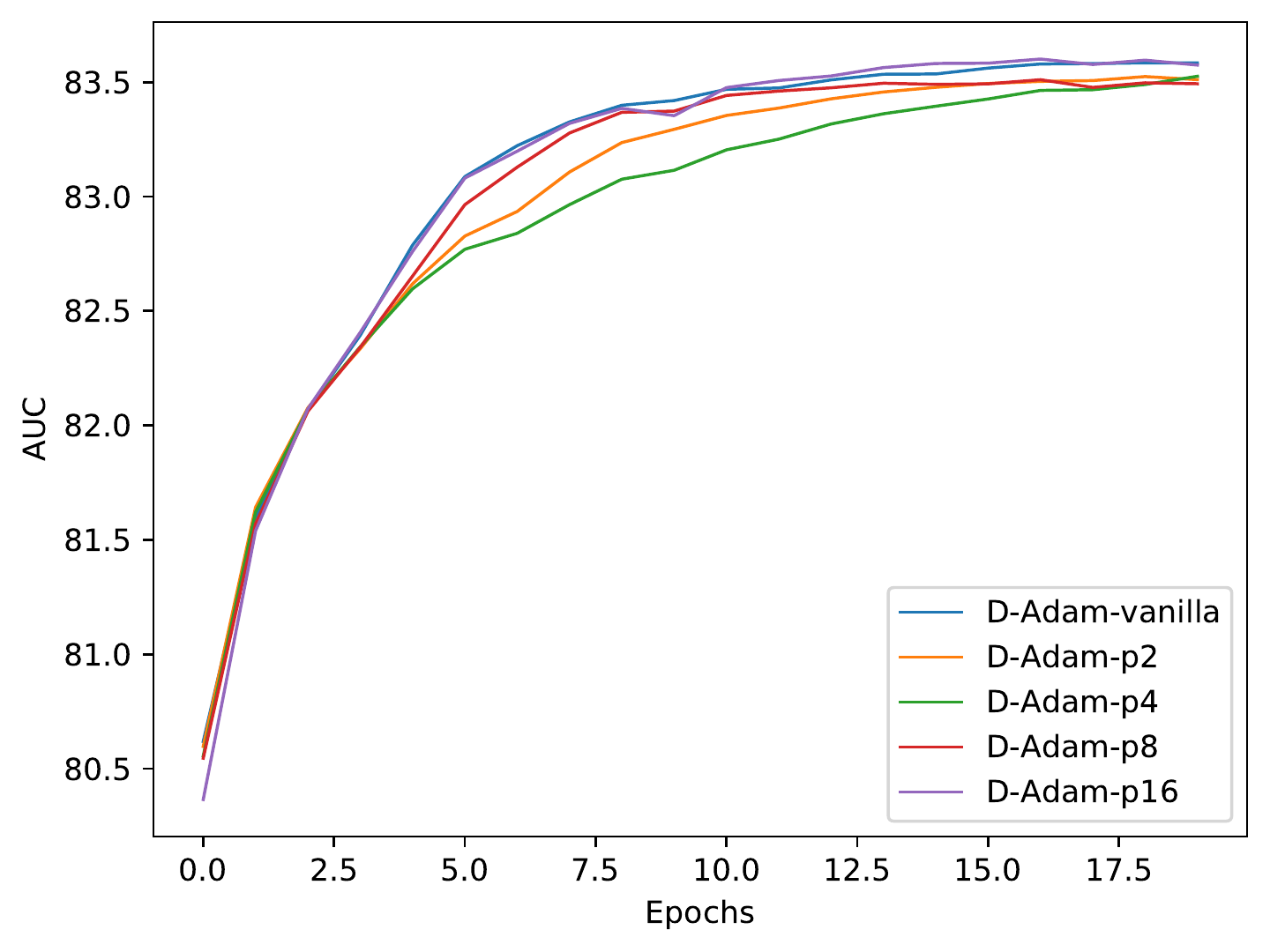}
		\label{test_loss_adam_avazu_epch}
	}
	\caption{ The testing performance of D-Adam w.r.t. communication cost (MB). }
	\label{convergence_curve_test_adam_epoch}
\end{figure}

	\begin{figure}[!htbp]
	\vspace{-3pt}
	\centering 
	\subfigure[CIFAR10]{
		\includegraphics[width=0.258\textwidth]{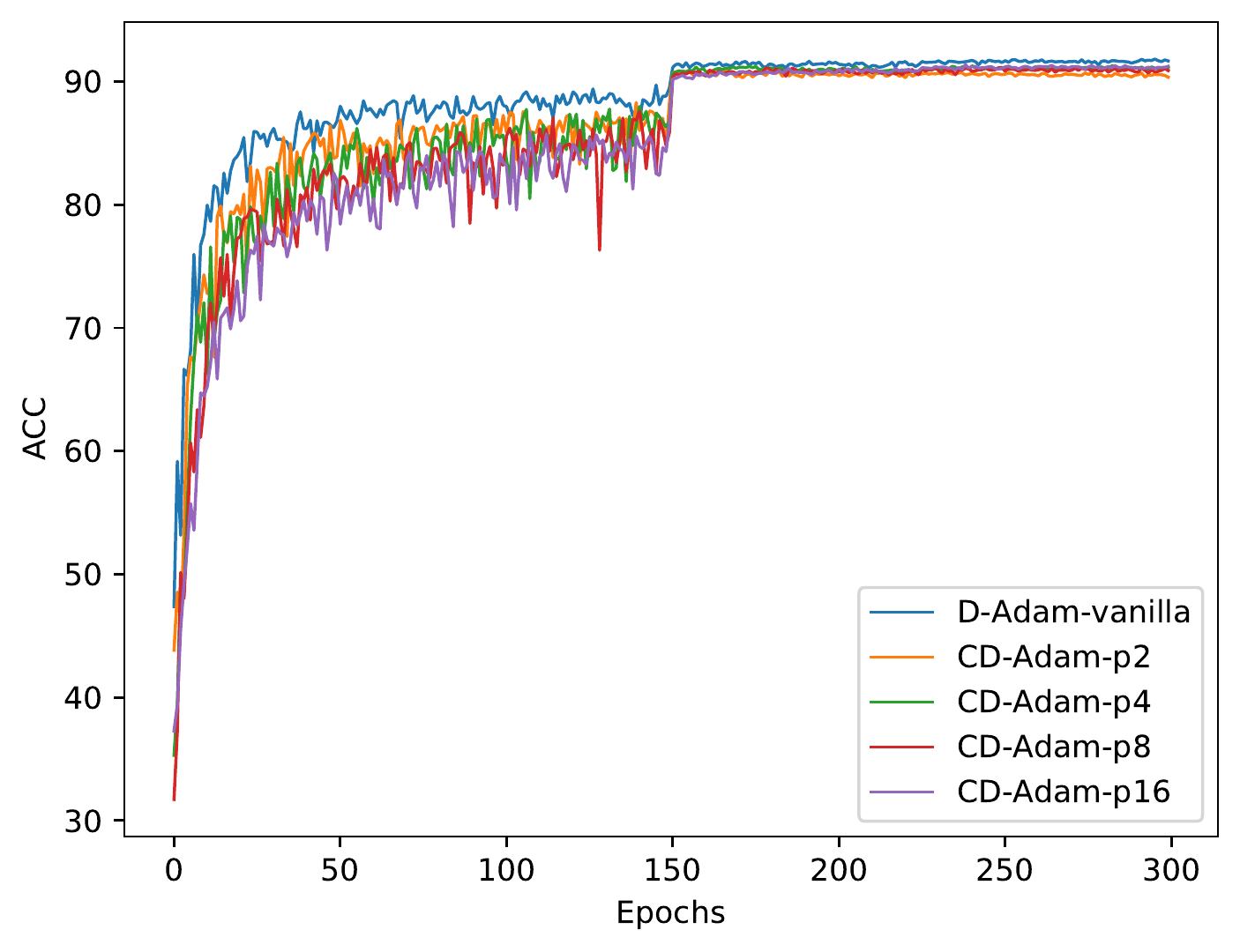}
		\label{test_acc_choco_adam_cifar10_epoch}
	}
	\qquad
	\subfigure[Criteo]{
		\includegraphics[width=0.259\textwidth]{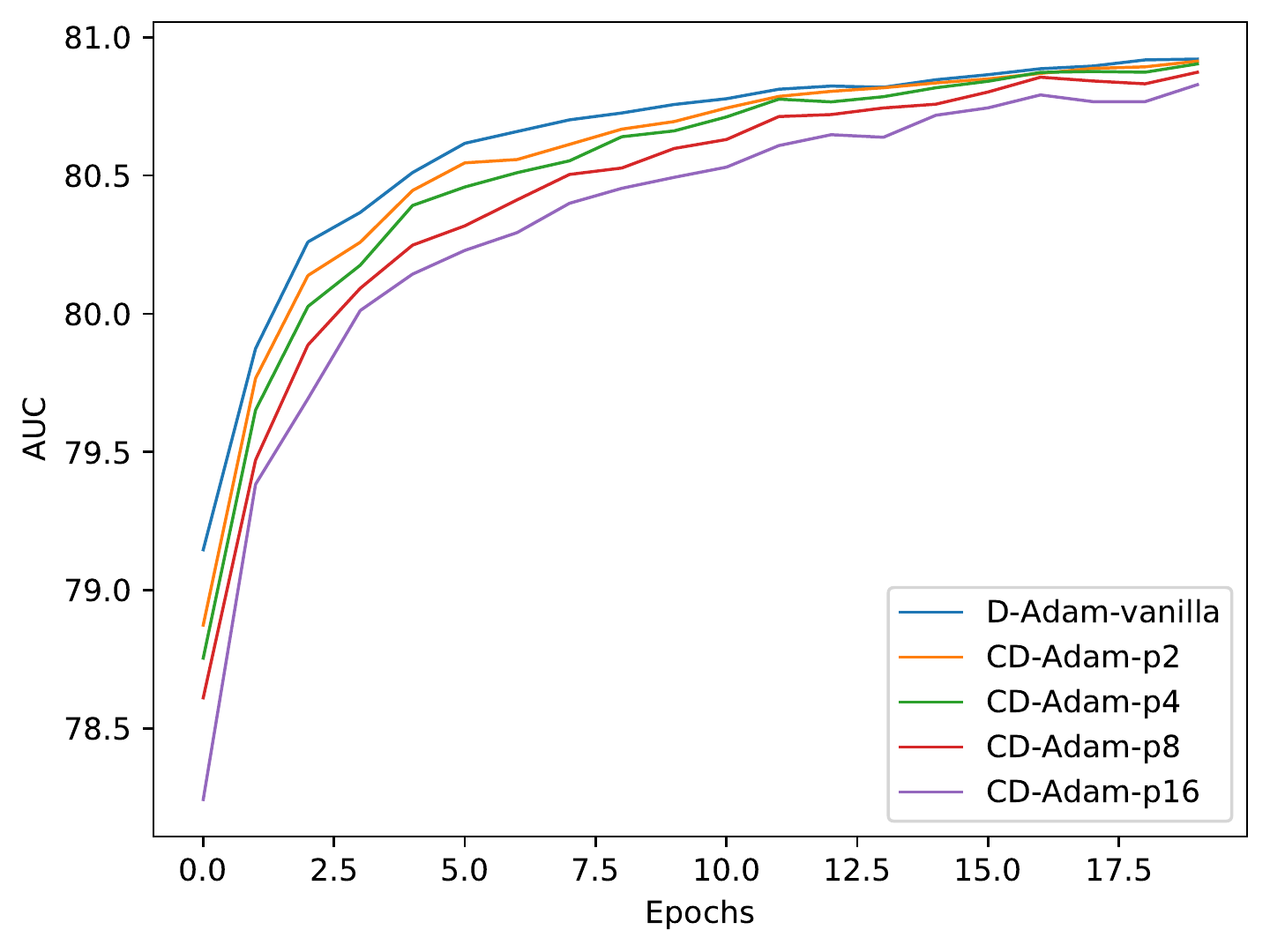}
		\label{test_loss_choco_adam_criteo_epoch}
	}
	\qquad
	\subfigure[Movielens-20M]{
		\includegraphics[width=0.258\textwidth]{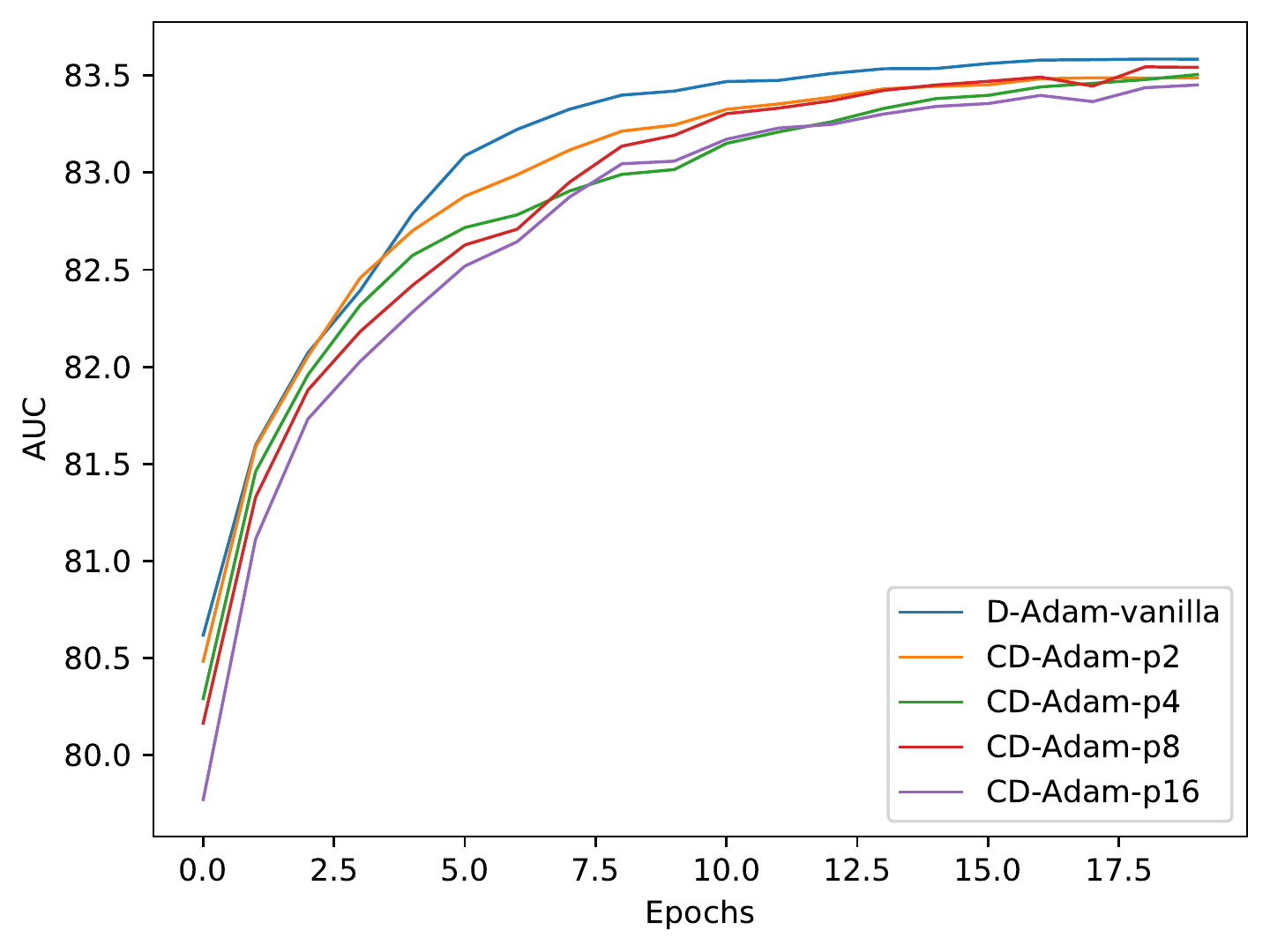}
		\label{test_loss_choco_adam_avazu_epoch}
	}
	\caption{ The testing performance of CD-Adam w.r.t. communication cost (MB). }
	\label{convergence_curve_test_choco_adam_epoch}
	\vspace{-3pt}
\end{figure}

%% file: sec_supp_proof.tex
\subsection{Proof of Theorem 1} \label{proof_theorem1}

To prove the convergence, we introduce  the following matrix form notations:
\begin{equation}
\begin{aligned}
& X_t = [\mathbf{x}_{t}^{(1)}, \mathbf{x}_{t}^{(2)}, \cdots, \mathbf{x}_{t}^{(K)}] \in \mathbb{R}^{d\times K}  \ ,\\
& \bar{X}_t = [\bar{\mathbf{x}}_t, \bar{\mathbf{x}}_t, \cdots, \bar{\mathbf{x}}_t] \in \mathbb{R}^{d\times K} \ , \\
& \Delta_t= [\frac{\mathbf{g}_t^{(1)}}{\sqrt{\mathbf{v}_{t}^{(1)} } + \tau} , \frac{\mathbf{g}_t^{(2)}}{\sqrt{\mathbf{v}_{t}^{(2)} } + \tau}  , \cdots, \frac{\mathbf{g}_t^{(k)}}{\sqrt{\mathbf{v}_{t}^{(k)} } + \tau} ] \in \mathbb{R}^{d\times K}   \ . \\
\end{aligned}
\end{equation}

Based on the aforementioned notations,  Algorithm~\ref{decentralized_adam} can be represented as follows:
\begin{equation} \label{mat_iter}
\begin{aligned}
& X_{t+\frac{1}{2}} = X_{t} - \eta \Delta_t   \ ,\\
& X_{t+1}=  X_{t+\frac{1}{2}}P  \ ,\\
\end{aligned}
\end{equation}
where $P=W$ when mod($t+1$, $p$)=0, otherwise $P=I$. In addition, we have
\begin{equation} \label{mat_avg}
\bar{X}_{t+1} =  X_{t+1}\frac{1}{K}\mathbf{1}\mathbf{1}^T   = \bar{X}_{t} - \eta\Delta_t \frac{1}{K}\mathbf{1}\mathbf{1}^T   \ ,
\end{equation}
where it follows from $P \frac{1}{K}\mathbf{1}\mathbf{1}^T= \frac{1}{K}\mathbf{1}\mathbf{1}^T$ since $W$ is a doubly stochastic matrix.
To prove Theorem~\ref{theorem1}, we introduce the following important lemmas.

\begin{lemma} \label{lemma_x_diff_1}
	Under Assumption~\ref{smooth}--\ref{norm}, we have
	\begin{equation}
	\mathbb{E}[\sum_{k=1}^{K}  \|\mathbf{x}_{t}^{(k)} - \bar{\mathbf{x}}_{t}\|^2 ]\leq  (1+\frac{4}{\rho^2})\frac{2d\eta^2p^2G^2K}{\tau^2} 
	\end{equation}
\end{lemma}

\begin{proof}
	
	Denoting $s_t=\floor{\frac{t}{p}}$, due to 
	\begin{equation}
	\begin{aligned}
	& X_{t} =  X_{s_tp} -  \eta \sum_{t'=s_tp}^{t-1}\Delta_{t'}   \ , \bar{X}_{t} = \bar{X}_{s_tp} - \eta \sum_{t'=s_tp}^{t-1} \Delta_{t'} \frac{1}{K}\mathbf{1}\mathbf{1}^T   \ ,\\
	\end{aligned}
	\end{equation}
	we have
	\begin{equation}
	\begin{aligned}
	&\mathbb{E}[ \| X_{t} - \bar{X}_{t} \|_F^2] = \mathbb{E}[\|X_{s_tp}- \bar{X}_{s_tp}+\eta \sum_{t'=s_tp}^{t-1}\Delta_{t'}( \frac{1}{K}\mathbf{1}\mathbf{1}^T  - \mathbf{I} )\|_F^2] \\
	& \leq  2\mathbb{E}[\|X_{s_tp}- \bar{X}_{s_tp} \|_F^2]+2\eta^2\mathbb{E}[\|\sum_{t'=s_tp}^{t-1}\Delta_{t'}( \frac{1}{K}\mathbf{1}\mathbf{1}^T  - \mathbf{I} )\|_F^2 ] \ .\\
	\end{aligned}
	\end{equation}
	Furthermore, we have

	\begin{equation}
	\begin{aligned}
	& \mathbb{E}[\|X_{s_tp} - \bar{X}_{s_tp}\|_F^2] \\
	& =  \mathbb{E}[\|(X_{(s_tp-1)+\frac{1}{2}} - \bar{X}_{(s_tp-1)+\frac{1}{2}})W\|_F^2] \\
	& =   \mathbb{E}[\|(X_{(s_tp-1)+\frac{1}{2}} - \bar{X}_{(s_tp-1)+\frac{1}{2}})(W-\frac{1}{K}\mathbf{1}\mathbf{1}^T)\|_F^2] \\
	& \leq \mathbb{E}[\|W-\frac{1}{K}\mathbf{1}\mathbf{1}^T\|_2^2  \|X_{(s_tp-1)+\frac{1}{2}} - \bar{X}_{(s_tp-1)+\frac{1}{2}}\|_F^2] \\
	& \leq (1-\rho)  \mathbb{E}[\|X_{(s_tp-1)+\frac{1}{2}} - \bar{X}_{(s_tp-1)+\frac{1}{2}}\|_F^2] \\
	%	& = (1-c) \|X_{(s_t-1)p}-\eta \Delta_t  - (\bar{X}_t - \eta\Delta_t \frac{1}{K}\mathbf{1}\mathbf{1}^T )\|_F^2 \\
	& = (1-\rho)  \mathbb{E}[\|X_{(s_t-1)p}- \bar{X}_{(s_t-1)p}+\eta \sum_{t'=(s_t-1)p}^{s_tp-1}\Delta_{t'}( \frac{1}{K}\mathbf{1}\mathbf{1}^T  - \mathbf{I} )\|_F^2 ]\\
	& \leq  (1-\rho)(1+\frac{1}{a})  \mathbb{E}[\|X_{(s_t-1)p}- \bar{X}_{(s_t-1)p}\|_F^2 ]+(1-\rho)(1+a)\eta^2  \mathbb{E}[\| \sum_{t'=(s_t-1)p}^{s_tp-1}\Delta_{t'} \|_F^2]\\
	& \leq (1-\frac{\rho}{2})\|X_{(s_t-1)p}- \bar{X}_{(s_t-1)p}\|_F^2  + \frac{2d\eta^2p^2G^2K}{\rho\tau^2}\\
	& \leq  \frac{4d\eta^2p^2G^2K}{\rho^2\tau^2}
	\end{aligned}
	\end{equation}
	where the second equality follows from $\bar{X}_{(s_tp-1)+\frac{1}{2}}=X_{(s_tp-1)+\frac{1}{2}}\frac{1}{K}\mathbf{1}\mathbf{1}^T$, the first inequality follows from $\|XY\|_F\leq \|X\|_2\|Y\|_F$, the second inequality follows from Lemma \ref{lemma_w_spectral} that $\|W-\frac{1}{K}\mathbf{1}\mathbf{1}^T\|_2^2\leq (1-\rho)^2\leq (1-\rho)$ since $\rho \in (0, 1]$, the second to last inequality follows from $a=\frac{2}{\rho}$, $\|\frac{\mathbf{g}_t^{(1)}}{\sqrt{\mathbf{v}_{t}^{(1)} } + \tau}\|^2\leq \frac{dG^2}{\tau^2}$, and the last step is obtained by recursive expansion. 
	
	Therefore, 
	\begin{equation}
	\begin{aligned}
	& \mathbb{E}[\sum_{k=1}^{K}  \|\mathbf{x}_{t}^{(k)} - \bar{\mathbf{x}}_{t}\|^2 ] = \mathbb{E}[ \| X_{t} - \bar{X}_{t} \|_F^2]   \leq \frac{8d\eta^2p^2G^2K}{\rho^2\tau^2} + \frac{2d\eta^2p^2G^2K}{\tau^2}  = (1+\frac{4}{\rho^2})\frac{2d\eta^2p^2G^2K}{\tau^2}  \ .
	\end{aligned}
	\end{equation}
	
\end{proof}

\vspace{-20pt}
Based on the aforementioned lemmas, we are ready to prove Theorem~\ref{theorem1}. Here, following  \cite{reddi2020adaptive,zaheer2018adaptive},  we consider $\beta_1=0$ for simplicity and it is easy to be extended to the general case. 
\begin{proof}
	From Eq.~(\ref{mat_avg}), we have
	\begin{equation}
	\bar{\mathbf{x}}_{t+1} = 	\bar{\mathbf{x}}_{t} - \frac{\eta }{K}\sum_{k=1}^{K}\frac{\mathbf{g}_t^{(k)}}{\sqrt{\mathbf{v}_{t}^{(k)} } + \tau}  \ .
	\end{equation}
	Based on the smoothness of the loss function, we have
	\begin{equation} \label{taylor_alg1}
	\begin{aligned}
	& \mathbb{E}[f(\bar{\mathbf{x}}_{t+1})] \leq  \mathbb{E}[f(\bar{\mathbf{x}}_{t}) + \langle \nabla  f(\bar{\mathbf{x}}_{t}), \bar{\mathbf{x}}_{t+1}-\bar{\mathbf{x}}_{t} \rangle + \frac{L}{2} \|\bar{\mathbf{x}}_{t+1}-\bar{\mathbf{x}}_{t} \|^2] \\
	& =  f(\bar{\mathbf{x}}_{t}) - \mathbb{E}[\langle \nabla  f(\bar{\mathbf{x}}_{t}), \frac{\eta }{K}\sum_{k=1}^{K}\frac{\mathbf{g}_t^{(k)}}{\sqrt{\mathbf{v}_{t}^{(k)} } + \tau}  \rangle] + \frac{L}{2} \mathbb{E}[\|\frac{\eta }{K}\sum_{k=1}^{K}\frac{\mathbf{g}_t^{(k)}}{\sqrt{\mathbf{v}_{t}^{(k)} } + \tau} \|^2] \\
	&  \leq f(\bar{\mathbf{x}}_{t}) - \mathbb{E}[\langle \nabla  f(\bar{\mathbf{x}}_{t}), \frac{\eta }{K}\sum_{k=1}^{K}\frac{\mathbf{g}_t^{(k)}}{\sqrt{\mathbf{v}_{t}^{(k)} } + \tau}  - \frac{\eta }{K}\sum_{k=1}^{K}\frac{\mathbf{g}_t^{(k)}}{\sqrt{\beta_2\mathbf{v}_{t-1}^{(k)} } + \tau} + \frac{\eta }{K}\sum_{k=1}^{K}\frac{\mathbf{g}_t^{(k)}}{\sqrt{\beta_2\mathbf{v}_{t-1}^{(k)} } + \tau} \rangle]  \\
	& \quad + \frac{L}{2} \mathbb{E}[\|\frac{\eta }{K}\sum_{k=1}^{K}\frac{\mathbf{g}_t^{(k)}}{\sqrt{\mathbf{v}_{t}^{(k)} } + \tau} \|^2] \\
	& \leq  f(\bar{\mathbf{x}}_{t}) \underbrace{- \frac{\eta }{K}\mathbb{E}[\sum_{k=1}^{K} \langle \nabla  f(\bar{\mathbf{x}}_{t}), \frac{\mathbf{g}_t^{(k)}}{\sqrt{\beta_2\mathbf{v}_{t-1}^{(k)} } + \tau} \rangle]}_{T_1}  \underbrace{ - \frac{\eta }{K}\mathbb{E}[\sum_{k=1}^{K} \langle \nabla  f(\bar{\mathbf{x}}_{t}),\frac{\mathbf{g}_t^{(k)}}{\sqrt{\mathbf{v}_{t}^{(k)} } + \tau}  -  \frac{\mathbf{g}_t^{(k)}}{\sqrt{\beta_2\mathbf{v}_{t-1}^{(k)} } + \tau} \rangle ] }_{T_2} \\
	& \quad + \frac{\eta^2L}{2} \mathbb{E}[\|\frac{1 }{K}\sum_{k=1}^{K}\frac{\mathbf{g}_t^{(k)}}{\sqrt{\mathbf{v}_{t}^{(k)} } + \tau} \|^2]  \ .\\
	\end{aligned}
	\end{equation}
	
	As for $T_1$ in Eq.~(\ref{taylor_alg1}), we have
	\begin{equation}
	\begin{aligned}
	& T_1=- \frac{\eta }{K}\mathbb{E}[\sum_{k=1}^{K} \langle \nabla  f(\bar{\mathbf{x}}_{t}), \frac{\mathbf{g}_t^{(k)}}{\sqrt{\beta_2\mathbf{v}_{t-1}^{(k)} } + \tau} \rangle ]\\
	&  = - \frac{\eta }{K}\mathbb{E}[\sum_{k=1}^{K} \langle \frac{\nabla  f(\bar{\mathbf{x}}_{t})}{\sqrt{\beta_2\mathbf{v}_{t-1}^{(k)} } + \tau}, \mathbf{g}_t^{(k)} - \nabla  f(\bar{\mathbf{x}}_{t}) + \nabla  f(\bar{\mathbf{x}}_{t})\rangle] \\
	& =  - \frac{\eta }{K}\mathbb{E}[\sum_{k=1}^{K}  \langle \frac{\nabla  f(\bar{\mathbf{x}}_{t})}{\sqrt{\beta_2\mathbf{v}_{t-1}^{(k)} } + \tau}, \nabla  f(\bar{\mathbf{x}}_{t})\rangle + \sum_{k=1}^{K}  \langle \frac{\nabla  f(\bar{\mathbf{x}}_{t})}{\sqrt{\beta_2\mathbf{v}_{t-1}^{(k)} } + \tau}, \mathbf{g}_t^{(k)} - \nabla  f(\bar{\mathbf{x}}_{t}) \rangle] \\
	& =  - \frac{\eta }{K} \mathbb{E}[\sum_{k=1}^{K}  \sum_{j=1}^{d}\frac{[\nabla  f(\bar{\mathbf{x}}_{t})]_j^2}{\sqrt{\beta_2\mathbf{v}_{t-1, j}^{(k)} } + \tau} ]- \frac{\eta }{K}\mathbb{E}[\sum_{k=1}^{K} \langle \frac{\nabla  f(\bar{\mathbf{x}}_{t})}{\sqrt{\beta_2\mathbf{v}_{t-1}^{(k)} } + \tau}, \mathbf{g}_t^{(k)} - \nabla  f(\bar{\mathbf{x}}_{t}) \rangle] \ .
	\end{aligned}
	\end{equation}
	
	Now, we can bound the last term in the last step as follows:
	\begin{equation}
	\begin{aligned}
	& - \frac{\eta }{K}\mathbb{E}[\sum_{k=1}^{K} \langle \frac{\nabla  f(\bar{\mathbf{x}}_{t})}{\sqrt{\beta_2\mathbf{v}_{t-1}^{(k)} } + \tau}, \mathbf{g}_t^{(k)} - \nabla  f(\bar{\mathbf{x}}_{t}) \rangle] \\
	& = - \frac{\eta }{K}\mathbb{E}[\sum_{k=1}^{K} \langle \frac{\nabla  f(\bar{\mathbf{x}}_{t})}{\sqrt{\sqrt{\beta_2\mathbf{v}_{t-1}^{(k)} } + \tau}}, \frac{\nabla  f^{(k)}({\mathbf{x}}_{t}^{(k)}) - \nabla  f(\bar{\mathbf{x}}_{t}) }{\sqrt{\sqrt{\beta_2\mathbf{v}_{t-1}^{(k)} } + \tau}}\rangle] \\
	& \leq   \frac{\eta }{2K} \sum_{k=1}^{K}  \sum_{j=1}^{d}\frac{[\nabla  f(\bar{\mathbf{x}}_{t})]_j^2}{\sqrt{\beta_2\mathbf{v}_{t-1, j}^{(k)} } + \tau} + \frac{\eta }{2K} \mathbb{E}[\sum_{k=1}^{K}  \|\frac{\nabla  f^{(k)}({\mathbf{x}}_{t}^{(k)}) - \nabla  f(\bar{\mathbf{x}}_{t}) }{\sqrt{\sqrt{\beta_2\mathbf{v}_{t-1}^{(k)} } + \tau}}\|^2 ]\\
	& \leq \frac{\eta }{2K} \sum_{k=1}^{K}  \sum_{j=1}^{d}\frac{[\nabla  f(\bar{\mathbf{x}}_{t})]_j^2}{\sqrt{\beta_2\mathbf{v}_{t-1, j}^{(k)} } + \tau} + \frac{\eta }{2\tau K} \mathbb{E}[\sum_{k=1}^{K}  \|\nabla  f^{(k)}({\mathbf{x}}_{t}^{(k)})   - \nabla  f(\bar{\mathbf{x}}_{t})  \|^2 ]\\
	& \leq \frac{\eta }{2K} \sum_{k=1}^{K}  \sum_{j=1}^{d}\frac{[\nabla  f(\bar{\mathbf{x}}_{t})]_j^2}{\sqrt{\beta_2\mathbf{v}_{t-1, j}^{(k)} } + \tau} + \frac{\eta L^2}{2\tau K} \mathbb{E}[\sum_{k=1}^{K}  \|\mathbf{x}_{t}^{(k)} - \bar{\mathbf{x}}_{t}\|^2 ]  \\
	& \leq \frac{\eta }{2K} \sum_{k=1}^{K}  \sum_{j=1}^{d}\frac{[\nabla  f(\bar{\mathbf{x}}_{t})]_j^2}{\sqrt{\beta_2\mathbf{v}_{t-1, j}^{(k)} } + \tau} +   (1+\frac{4}{\rho^2})\frac{d\eta^3p^2G^2L^2}{\tau^3}   \ ,\\
	\end{aligned}
	\end{equation}
	where the second to last inequality follows from Assumption~\ref{smooth}, the last inequality follows from Lemma~\ref{lemma_x_diff_1}.
	
	As for $T_2$ in Eq.~(\ref{taylor_alg1}), we have
	\begin{equation}
	\begin{aligned}
	&  T_2=- \frac{\eta }{K}\mathbb{E}[\sum_{k=1}^{K} \langle \nabla  f(\bar{\mathbf{x}}_{t}),\frac{\mathbf{g}_t^{(k)}}{\sqrt{\mathbf{v}_{t}^{(k)} } + \tau}   -  \frac{\mathbf{g}_t^{(k)}}{\sqrt{\beta_2\mathbf{v}_{t-1}^{(k)} } + \tau} \rangle]  \\
	&  = -\frac{\eta }{K}\mathbb{E}[\sum_{k=1}^{K} \sum_{j=1}^{d} [ \nabla  f(\bar{\mathbf{x}}_{t})]_j\times \mathbf{g}_{t,j}^{(k)}\times (\frac{1}{\sqrt{\mathbf{v}_{t,j}^{(k)}}+\tau} - \frac{1}{\sqrt{\beta_2\mathbf{v}_{t-1,j}^{(k)}}+\tau}) ]\\
	&  = -\frac{\eta }{K}\mathbb{E}[\sum_{k=1}^{K} \sum_{j=1}^{d} [ \nabla  f(\bar{\mathbf{x}}_{t})]_j\times \mathbf{g}_{t,j}^{(k)}\times \frac{\sqrt{\beta_2\mathbf{v}_{t-1,j}^{(k)}}- \sqrt{\mathbf{v}_{t,j}^{(k)}}}{(\sqrt{\mathbf{v}_{t,j}^{(k)}}+\tau)(\sqrt{\beta_2\mathbf{v}_{t-1,j}^{(k)}}+\tau)} ]\\
	& =  -\frac{\eta }{K}\mathbb{E}[\sum_{k=1}^{K} \sum_{j=1}^{d} [ \nabla  f(\bar{\mathbf{x}}_{t})]_j\times \mathbf{g}_{t,j}^{(k)}\times \frac{{\beta_2\mathbf{v}_{t-1,j}^{(k)}}- {\mathbf{v}_{t,j}^{(k)}}}{(\sqrt{\mathbf{v}_{t,j}^{(k)}}+\tau)(\sqrt{\beta_2\mathbf{v}_{t-1,j}^{(k)}}+\tau)(\sqrt{\beta_2\mathbf{v}_{t-1,j}^{(k)}}+ \sqrt{\mathbf{v}_{t,j}^{k}})}] \\
	& =  \frac{\eta }{K}\mathbb{E}[\sum_{k=1}^{K} \sum_{j=1}^{d} [ \nabla  f(\bar{\mathbf{x}}_{t})]_j\times \mathbf{g}_{t,j}^{(k)}\times \frac{(1-\beta_2)(\mathbf{g}_{t,j}^{(k)})^2}{(\sqrt{\mathbf{v}_{t,j}^{(k)}}+\tau)(\sqrt{\beta_2\mathbf{v}_{t-1,j}^{(k)}}+\tau)(\sqrt{\beta_2\mathbf{v}_{t-1,j}^{(k)}}+ \sqrt{\mathbf{v}_{t,j}^{k}})} ]\\
	& \leq \frac{\eta }{K}\mathbb{E}[\sum_{k=1}^{K} \sum_{j=1}^{d} |[ \nabla  f(\bar{\mathbf{x}}_{t})]_j|\times |\mathbf{g}_{t,j}^{(k)}|\times \frac{(1-\beta_2)(\mathbf{g}_{t,j}^{(k)})^2}{(\sqrt{\mathbf{v}_{t,j}^{(k)}}+\tau)(\sqrt{\beta_2\mathbf{v}_{t-1,j}^{(k)}}+\tau)(\sqrt{\beta_2\mathbf{v}_{t-1,j}^{(k)}}+ \sqrt{\mathbf{v}_{t,j}^{k}})} ]\\
	& \leq \frac{\eta\sqrt{1-\beta_2} }{K}\mathbb{E}[\sum_{k=1}^{K} \sum_{j=1}^{d} |[ \nabla  f(\bar{\mathbf{x}}_{t})]_j|\times \frac{(\mathbf{g}_{t,j}^{(k)})^2}{(\sqrt{\mathbf{v}_{t,j}^{(k)}}+\tau)(\sqrt{\beta_2\mathbf{v}_{t-1,j}^{(k)}}+\tau)} ]\\
	& \leq \frac{\eta G\sqrt{1-\beta_2} }{\tau K}\mathbb{E}[\sum_{k=1}^{K} \sum_{j=1}^{d}  \frac{(\mathbf{g}_{t,j}^{(k)})^2}{\sqrt{\mathbf{v}_{t,j}^{(k)}}+\tau} ] \ ,\\
	\end{aligned}
	\end{equation}
	where the fourth equality and the second inequality follow from $\mathbf{v}_{t}^{(k)} = \beta_2\mathbf{v}_{t-1}^{(k)} + (1-\beta_2)\mathbf{g}_t^{(k)}\circ \mathbf{g}_t^{(k)}$ in Algorithm~\ref{decentralized_adam}, 
	the last inequality follows that $\sqrt{\beta_2\mathbf{v}_{t-1,j}^{(k)}}+\tau>\tau$ and $|[ \nabla  f(\bar{\mathbf{x}}_{t})]_j|<G$.

	Putting $T_1$ and $T_2$ into Eq.~(\ref{taylor_alg1}), we have
	\begin{equation}
	\begin{aligned}
	&  \mathbb{E}[f(\bar{\mathbf{x}}_{t+1}) ] \leq  f(\bar{\mathbf{x}}_{t}) -\frac{\eta }{2K} \sum_{k=1}^{K}  \sum_{j=1}^{d}\frac{[\nabla  f(\bar{\mathbf{x}}_{t})]_j^2}{\sqrt{\beta_2\mathbf{v}_{t-1, j}^{(k)} } + \tau} + (1+\frac{4}{\rho^2})\frac{d\eta^3p^2G^2L^2}{\tau^3}  \\
	& + \frac{\eta G\sqrt{1-\beta_2} }{\tau K}\mathbb{E}[\sum_{k=1}^{K} \sum_{j=1}^{d}  \frac{(\mathbf{g}_{t,j}^{(k)})^2}{\sqrt{\mathbf{v}_{t,j}^{(k)}}+\tau}]  + \frac{\eta^2L}{2K} \mathbb{E}[\sum_{k=1}^{K} \sum_{j=1}^{d} \frac{(\mathbf{g}_{t,j}^{(k)})^2}{\mathbf{v}_{t,j}^{(k)}+\tau^2} ]  \ .\\
	\end{aligned}
	\end{equation}

	Furthermore, 
	\begin{equation}
	\begin{aligned}
	& \frac{1}{K}\sum_{k=1}^{K} \sum_{j=1}^{d} \frac{(\mathbf{g}_{t,j}^{(k)})^2}{\mathbf{v}_{t,j}^{(k)}+\tau^2}  \leq \frac{1}{K}\sum_{k=1}^{K} \sum_{j=1}^{d} \frac{(\mathbf{g}_{t,j}^{(k)})^2}{\tau^2} \\
	& = \frac{1}{\tau^2K}\sum_{k=1}^{K} \|\mathbf{g}_{t}^{(k)} - \nabla f^{(k)}(\mathbf{x}_t^{(k)}) +  \nabla f^{(k)}(\mathbf{x}_t^{(k)})  - \nabla f(\bar{\mathbf{x}}_t) +  \nabla f(\bar{\mathbf{x}}_t) \|^2 \\
	& \leq \frac{3}{\tau^2K}\sum_{k=1}^{K}\Big(\|\mathbf{g}_{t}^{(k)} - \nabla f^{(k)}(\mathbf{x}_t^{(k)})\|^2 +\| \nabla f^{(k)}(\mathbf{x}_t^{(k)}) - \nabla f(\bar{\mathbf{x}}_t) \|^2 +\|\nabla f(\bar{\mathbf{x}}_t) \|^2\Big)  \\
	& \leq \frac{3}{\tau^2}\sum_{j=1}^d\sigma_{j}^2+ \frac{3}{\tau^2K}\sum_{k=1}^{K} \|\mathbf{x}_t^{(k)} - \bar{\mathbf{x}}\|^2 + \frac{3}{\tau^2}  \|\nabla f(\bar{\mathbf{x}}_t) \|^2 \\
	& \leq\frac{3}{\tau^2}\sum_{j=1}^d\sigma_{j}^2+ (1+ \frac{4}{\rho^2}) \frac{6d\eta^2p^2G^2}{\tau^4}  + \frac{3}{\tau^2}  \|\nabla f(\bar{\mathbf{x}}_t) \|^2 \ ,
	\end{aligned}
	\end{equation}
	where the last inequality follows from Lemma~\ref{lemma_x_diff_1}.
	Similarly, we have
	\begin{equation}
	\begin{aligned}
	& \frac{1}{\tau K}\sum_{k=1}^{K} \sum_{j=1}^{d} \frac{(\mathbf{g}_{t,j}^{(k)})^2}{\sqrt{\mathbf{v}_{t,j}^{(k)}}+\tau}  \leq \frac{1}{K}\sum_{k=1}^{K} \sum_{j=1}^{d} \frac{(\mathbf{g}_{t,j}^{(k)})^2}{\tau^2} \leq\frac{3}{\tau^2}\sum_{j=1}^d\sigma_{j}^2+ (1+ \frac{4}{\rho^2}) \frac{6d\eta^2p^2G^2}{\tau^4}  + \frac{3}{\tau^2}  \|\nabla f(\bar{\mathbf{x}}_t) \|^2
	\end{aligned}
	\end{equation}

	Then, we have
	\begin{equation}
	\begin{aligned}
	& \frac{1}{T}\sum_{t=0}^{T-1}\frac{\eta }{2K} \sum_{k=1}^{K}  \sum_{j=1}^{d}\frac{[\nabla  f(\bar{\mathbf{x}}_{t})]_j^2}{\sqrt{\beta_2\mathbf{v}_{t-1, j}^{(k)} } + \tau} \leq \frac{f(\mathbf{x}_0) -   f_*}{T}+  (1+ \frac{4}{\rho^2}) \frac{d\eta^3p^2G^2L^2}{\tau^3} \\
	& + (\eta G\sqrt{1-\beta_2} + \frac{\eta^2L}{2})\Big(\frac{3}{\tau^2}\sum_{j=1}^d\sigma_{j}^2+ (1+ \frac{4}{\rho^2}) \frac{6d\eta^2p^2G^2}{\tau^4}  + \frac{3}{\tau^2T} \sum_{t=0}^{T-1}  \|\nabla f(\bar{\mathbf{x}}_t) \|^2\Big)  \ .\\
	\end{aligned}
	\end{equation}
	By setting $\eta<\frac{\tau^2}{3\sqrt{\beta_2}GL}$, we have 
	\begin{equation}
	\begin{aligned}
	& \frac{1}{T}\sum_{t=0}^{T-1}\frac{\eta }{2K} \sum_{k=1}^{K}  \sum_{j=1}^{d}\frac{[\nabla  f(\bar{\mathbf{x}}_{t})]_j^2}{\sqrt{\beta_2\mathbf{v}_{t-1, j}^{(k)} } + \tau} \leq \frac{f(\mathbf{x}_0) -   f_* }{T}+  (1+ \frac{4}{\rho^2}) \frac{d\eta^3p^2G^2L^2}{\tau^3} \\
	& + (\eta G\sqrt{1-\beta_2} + \frac{\eta^2L}{2})\Big(\frac{3}{\tau^2}\sum_{j=1}^d\sigma_{j}^2+ (1+ \frac{4}{\rho^2}) \frac{6d\eta^2p^2G^2}{\tau^4}  \Big)  \ .\\
	\end{aligned}
	\end{equation}
	In addition,  by setting $0<\tau<1$, we have 
	\begin{equation}
	\begin{aligned}
	& \frac{1}{T}\sum_{t=0}^{T-1}\frac{\eta }{2K} \sum_{k=1}^{K}  \sum_{j=1}^{d}\frac{[\nabla  f(\bar{\mathbf{x}}_{t})]_j^2}{\sqrt{\beta_2\mathbf{v}_{t-1, j}^{(k)} } + \tau} \geq \frac{1}{T}\sum_{t=0}^{T-1}\frac{\eta }{2K} \sum_{k=1}^{K}  \sum_{j=1}^{d}\frac{[\nabla  f(\bar{\mathbf{x}}_{t})]_j^2}{\sqrt{\beta_2}G + \tau} = \frac{\eta}{2T(\sqrt{\beta_2}G + 1)}\sum_{t=0}^{T-1} \|\nabla  f(\bar{\mathbf{x}}_{t})\|^2 \ .
	\end{aligned}
	\end{equation}
	At last, we have
	\begin{equation}
	\begin{aligned}
	& \frac{1}{T}\sum_{t=0}^{T-1} \|\nabla  f(\bar{\mathbf{x}}_{t})\|^2 \leq 2(\sqrt{\beta_2}G + 1)\Bigg(\frac{f(\mathbf{x}_0) -   f_*}{\eta T}+  (1+ \frac{4}{\rho^2}) \frac{d\eta^2p^2G^2L^2}{\tau^3} \\
	& + ( G\sqrt{1-\beta_2} + \frac{\eta L}{2})\Big(\frac{3}{\tau^2}\sum_{j=1}^d\sigma_{j}^2+ (1+ \frac{4}{\rho^2}) \frac{6d\eta^2p^2G^2}{\tau^4} \Big)\Bigg)  \ . \\
	\end{aligned}
	\end{equation}
	
\end{proof}

\subsection{Proof of Theorem 2}  \label{proof_theorem2}
Based on the aforementioned matrix notations, when mod($t+1$, $p$)=0,  Algorithm~\ref{decentralized_adam_com} can be represented as follows:
\begin{equation} \label{matrix_compress}
\begin{aligned}
& X_{(s_tp-1)+\frac{1}{2}} = X_{(s_t-1)p} -\eta \sum_{t'=(s_t-1)p}^{s_tp-1}\Delta_{t'}\\
& X_{s_tp} = X_{(s_tp-1)+\frac{1}{2}} +\gamma \hat{X}_{(s_t-1)p}(W-I) \\ 
& \hat{X}_{s_tp} = \hat{X}_{(s_t-1)p} + Q(X_{s_tp}- \hat{X}_{(s_t-1)p})
\end{aligned}
\end{equation}
where $s_t=\floor{\frac{t+1}{p}}$ and $\hat{X}_t = [\hat{\mathbf{x}}_t^{(1)}, \hat{\mathbf{x}}_t^{(2)}, \cdots, \hat{\mathbf{x}}_t^{(K)}] \in \mathbb{R}^{d\times K} $.  In addition, we have
\begin{equation} \label{x_bar_compress}
\begin{aligned}
& \bar{X}_{s_tp} =  \bar{X}_{(s_tp-1)+\frac{1}{2}} +\gamma \hat{X}_{(s_t-1)p}(W-I)\frac{1}{K}\mathbf{1}\mathbf{1}^T =  \bar{X}_{(s_tp-1)+\frac{1}{2}} 
\end{aligned}
\end{equation}
where the last step follows that $W$ is a doubly stochastic matrix. 

\begin{lemma} \label{lemma_x_diff_2}
		Under Assumption~\ref{smooth}--\ref{norm}, we have
	\begin{equation}
	\begin{aligned}
	&\mathbb{E}[\sum_{k=1}^{K}  \|\mathbf{x}_{t}^{(k)} - \bar{\mathbf{x}}_{t}\|^2 ] \leq \frac{8d\eta^2 p^2G^2K}{\tau^2}(1+\frac{2}{\alpha^2}) \ ,\\
	\end{aligned}
	\end{equation}
	where $\alpha=\frac{\rho^2\delta}{82}$.
\end{lemma}
\begin{proof}
	Denoting $s_t=\floor{\frac{t}{p}}$, in terms of the updating rules in Algorithm~\ref{decentralized_adam_com}, we have
	\begin{equation}
	\begin{aligned}
	&\mathbb{E}[ \| X_{t} - \bar{X}_{t} \|_F^2] \\
	& \leq \mathbb{E}[ \| X_{t} - \bar{X}_{t} \|_F^2] +  \mathbb{E}[ \| X_{t} - \hat{X}_{t} \|_F^2] \\
	&= \mathbb{E}[\|X_{s_tp}- \bar{X}_{s_tp}+\eta \sum_{t'=s_tp}^{t-1}\Delta_{t'}( \frac{1}{K}\mathbf{1}\mathbf{1}^T  - \mathbf{I} )\|_F^2]  + \mathbb{E}[\|X_{s_tp} - \hat{X}_{s_tp} - \eta \sum_{t'=s_tp}^{t-1}\Delta_{t'}\|_F^2]  \\
	& \leq  2\mathbb{E}[\|X_{s_tp}- \bar{X}_{s_tp} \|_F^2]+ 2  \mathbb{E}[\|X_{s_tp} - \hat{X}_{s_tp} \|_F^2]\\
	& \quad +2\eta^2\mathbb{E}[\|\sum_{t'=s_tp}^{t}\Delta_{t'}( \frac{1}{K}\mathbf{1}\mathbf{1}^T  - \mathbf{I} )\|_F^2 ]  +  2\eta^2\mathbb{E}[\| \sum_{t'=s_tp}^{t}\Delta_{t'}\|_F^2]  \\
	& \leq 2\mathbb{E}[\|X_{s_tp}- \bar{X}_{s_tp} \|_F^2]+ 2  \mathbb{E}[\|X_{s_tp} - \hat{X}_{s_tp} \|_F^2] +  \frac{8\eta^2dp^2G^2K}{\tau^2} \ .\\
	\end{aligned}
	\end{equation}
	
	Now, we will bound the first two terms by following \cite{koloskova2019decentralizedcon}.   As for the first term, we have
	\begin{equation}
	\begin{aligned}
	& \|X_{s_tp} - \bar{X}_{s_tp}\|_F^2  \\
	& =  \| X_{(s_tp-1)+\frac{1}{2}} +\gamma \hat{X}_{(s_t-1)p}(W-I) -\bar{X}_{(s_tp-1)+\frac{1}{2}} \|_F^2 \\
	& = \| X_{(s_tp-1)+\frac{1}{2}} +\gamma \hat{X}_{(s_t-1)p}(W-I) -\bar{X}_{(s_tp-1)+\frac{1}{2}} - \gamma\bar{X}_{(s_tp-1)+\frac{1}{2}} (W-I)\|_F^2 \\
	& = \| (X_{(s_tp-1)+\frac{1}{2}}- \bar{X}_{(s_tp-1)+\frac{1}{2}})(I+\gamma(W-I)) + \gamma(\hat{X}_{(s_t-1)p} - X_{(s_tp-1)+\frac{1}{2}})(W-I)\|_F^2\\
	& \leq (1+c_1) \|(X_{(s_tp-1)+\frac{1}{2}}- \bar{X}_{(s_tp-1)+\frac{1}{2}})(I+\gamma(W-I)) \|_F^2 + (1+\frac{1}{c_1}) \| \gamma(\hat{X}_{(s_t-1)p} - X_{(s_tp-1)+\frac{1}{2}})(W-I)\|_F^2 \\
	& \leq (1+c_1)(1-\gamma \rho)^2\|X_{(s_tp-1)+\frac{1}{2}}- \bar{X}_{(s_tp-1)+\frac{1}{2}} \|_F^2 + (1+\frac{1}{c_1})\gamma^2\beta^2\| \hat{X}_{(s_t-1)p} - X_{(s_tp-1)+\frac{1}{2}}\|_F^2  \ ,
	\end{aligned}
	\end{equation}
	where $\beta=\max_i\{1-\lambda_i\}$, the first equality follows from Eq.~(\ref{matrix_compress}) and Eq.~(\ref{x_bar_compress}). The second step follows that $\bar{X}_{(s_tp-1)+\frac{1}{2}} (W-I)=0$. 
	The last step follows from $\|XW\|_F\leq \|X\|_F\|W\|_2$ and
	\begin{equation}
	\begin{aligned}
	& \|(X_{(s_tp-1)+\frac{1}{2}}- \bar{X}_{(s_tp-1)+\frac{1}{2}})(I+\gamma(W-I)) \|_F\\
	& \leq (1-\gamma)\|X_{(s_tp-1)+\frac{1}{2}}- \bar{X}_{(s_tp-1)+\frac{1}{2}}\|_F + \gamma \|(X_{(s_tp-1)+\frac{1}{2}}- \bar{X}_{(s_tp-1)+\frac{1}{2}})W\|_F\\
	& =  (1-\gamma)\|X_{(s_tp-1)+\frac{1}{2}}- \bar{X}_{(s_tp-1)+\frac{1}{2}}\|_F + \gamma\|(X_{(s_tp-1)+\frac{1}{2}}- \bar{X}_{(s_tp-1)+\frac{1}{2}})(W-\frac{\mathbf{1}\mathbf{1}^T}{K})\|_F \\
	& \leq (1-\gamma\rho)\|(X_{(s_tp-1)+\frac{1}{2}}- \bar{X}_{(s_tp-1)+\frac{1}{2}})\|_F \ ,
	\end{aligned}
	\end{equation}
	where the first step follows from the convexity of Frobenius norm, the  second step follows that $ \bar{X}_{(s_tp-1)+\frac{1}{2}} = X_{(s_tp-1)+\frac{1}{2}}\frac{\mathbf{1}\mathbf{1}^T}{n}$ and $ \bar{X}_{(s_tp-1)+\frac{1}{2}} = {X}_{(s_tp-1)+\frac{1}{2}}\frac{\mathbf{1}\mathbf{1}^T}{n}$, the last step follows from Lemma~\ref{lemma_w_spectral}.
	
	As for the second term, we have
	\begin{equation}
	\begin{aligned}
	&  \|X_{s_tp} - \hat{X}_{s_tp}\|_F^2  \\
	& = \|X_{s_tp} - \hat{X}_{(s_t-1)p}-Q(X_{s_tp}- \hat{X}_{(s_t-1)p})\|_F^2 \\
	& \leq (1-\delta)\|X_{s_tp} - \hat{X}_{(s_t-1)p}\|_F^2 \\
	& =  (1-\delta) \|  X_{(s_tp-1)+\frac{1}{2}} +\gamma \hat{X}_{(s_t-1)p}(W-I) - \hat{X}_{(s_t-1)p}\|_F^2 \\
	& =  (1-\delta) \|  X_{(s_tp-1)+\frac{1}{2}} - \hat{X}_{(s_t-1)p}(I - \gamma(W-I)) -\gamma\bar{X}_{(s_tp-1)+\frac{1}{2}} (W-I)\|_F^2 \\
	& = (1-\delta)  \|(X_{(s_tp-1)+\frac{1}{2}} - \hat{X}_{(s_t-1)p})(I - \gamma(W-I)) +\gamma(X_{(s_tp-1)-\frac{1}{2}}-\bar{X}_{(s_tp-1)+\frac{1}{2}} )(W-I)\|_F^2 \\
	& \leq (1-\delta)(1+c_2)\|(X_{(s_tp-1)+\frac{1}{2}} - \hat{X}_{(s_t-1)p})(I - \gamma(W-I))\|_F^2 \\
	& \quad + (1-\delta)(1+\frac{1}{c_2})\|\gamma(X_{(s_tp-1)+\frac{1}{2}}-\bar{X}_{(s_tp-1)+\frac{1}{2}} )(W-I)\|_F^2 \\
	& \leq (1-\delta)(1+c_2)(1+\gamma\beta)^2\|X_{(s_tp-1)+\frac{1}{2}} - \hat{X}_{(s_t-1)p}\|_F^2 \\
	& \quad + (1-\delta)\gamma^2\beta^2(1+\frac{1}{c_2})\|X_{(s_tp-1)+\frac{1}{2}}-\bar{X}_{(s_tp-1)+\frac{1}{2}} \|_F^2  \ ,
	\end{aligned}
	\end{equation}
	where the second step follows from Definition~\ref{compress_op}, the third step follows from Eq.~(\ref{matrix_compress}), the fourth step follows that $\bar{X}_{(s_tp-1)+\frac{1}{2}} (W-I)=0$, the last step follows from $\|XW\|_F\leq \|X\|_F\|W\|_2$. 
	
	Combining these two terms, we have
	\begin{equation}
	\begin{aligned}
	& \|X_{s_tp} - \bar{X}_{s_tp}\|_F^2 + \|X_{s_tp} - \hat{X}_{s_tp}\|_F^2  \\
	& \leq ((1+c_1)(1-\gamma \rho)^2 + (1-\delta)\gamma^2\beta^2(1+\frac{1}{c_2}))\|X_{(s_tp-1)+\frac{1}{2}}- \bar{X}_{(s_tp-1)+\frac{1}{2}} \|_F^2 \\
	& \quad +( (1+\frac{1}{c_1})\gamma^2\beta^2 + (1-\delta)(1+c_2)(1+\gamma\beta)^2)\| \hat{X}_{(s_t-1)p} - X_{(s_tp-1)+\frac{1}{2}}\|_F^2 \ . 
	\end{aligned}
	\end{equation}
	Similar with \cite{koloskova2019decentralizedcon}, by setting $c_{1}=\frac{\gamma \rho}{2}$, $c_{2}=\frac{\delta}{2}$, and $\gamma=\frac{\rho \delta}{16 \rho+\rho^{2}+4 \beta^{2}+2 \rho \beta^{2}-8 \rho \delta}$, we have $\alpha=\frac{\rho^2\delta}{82}\in (0, 1)$ such that
	\begin{equation}
	\begin{aligned}
	& \mathbb{E}[\|X_{s_tp} - \bar{X}_{s_tp}\|_F^2]+ \mathbb{E}[\|X_{s_tp} - \hat{X}_{s_tp}\|_F^2] \\
	&  \leq (1-\alpha) \mathbb{E}[\|X_{(s_tp-1)+\frac{1}{2}}- \bar{X}_{(s_tp-1)+\frac{1}{2}} \|_F^2]  + (1-\alpha)\mathbb{E}[\| \hat{X}_{(s_t-1)p} - X_{(s_tp-1)+\frac{1}{2}}\|_F^2] \\
	& = (1-\alpha) \mathbb{E}[\|X_{(s_t-1)p}- \bar{X}_{(s_t-1)p}+\eta \sum_{t'=(s_t-1)p}^{s_tp-1}\Delta_{t'}( \frac{1}{K}\mathbf{1}\mathbf{1}^T  - \mathbf{I} )\|_F^2 ] \\
	& \quad + (1-\alpha) \mathbb{E}[\| \hat{X}_{(s_t-1)p} -X_{(s_t-1)p} +\eta \sum_{t'=(s_t-1)p}^{s_tp-1}\Delta_{t'}\|_F^2] \\
	& \leq (1-\alpha)\Big((1+\frac{1}{c})\mathbb{E}[\|X_{(s_t-1)p}- \bar{X}_{(s_t-1)p}\|_F^2 + \| X_{(s_t-1)p} -\hat{X}_{(s_t-1)p} \|_F^2]\\
	& \quad +(1+c)\eta^2 \mathbb{E}[\|\sum_{t'=(s_t-1)p}^{s_tp-1}\Delta_{t'}( \frac{1}{K}\mathbf{1}\mathbf{1}^T  - \mathbf{I} )\|_F^2+\|\sum_{t'=(s_t-1)p}^{s_tp-1}\Delta_{t'}\|_F^2]\Big) \\
	& \leq  (1-\alpha)\Big((1+\frac{1}{c})\mathbb{E}[\|X_{(s_t-1)p}- \bar{X}_{(s_t-1)p}\|_F^2 + \| X_{(s_t-1)p} -\hat{X}_{(s_t-1)p} \|_F^2] + 2\eta^2(1+c)\frac{dp^2G^2K}{\tau^2}\Big) \\
	& \leq (1-\frac{\alpha}{2})\Big(\mathbb{E}[\|X_{(s_t-1)p}- \bar{X}_{(s_t-1)p}\|_F^2 + \| X_{(s_t-1)p} -\hat{X}_{(s_t-1)p} \|_F^2] \Big)  + \frac{4d\eta^2p^2G^2K}{\alpha\tau^2}\\
	& \leq  \frac{8d\eta^2p^2G^2K}{\alpha^2\tau^2}
	\end{aligned}
	\end{equation}
	where the second to the last step follows from $c=2/\alpha$, the last step is obtained by recursive expansion. Therefore, we have
	\begin{equation}
	\begin{aligned}
	&\mathbb{E}[\sum_{k=1}^{K}  \|\mathbf{x}_{t}^{(k)} - \bar{\mathbf{x}}_{t}\|^2 ]  = \mathbb{E}[ \| X_{t} - \bar{X}_{t} \|_F^2] \leq  \frac{16d\eta^2p^2G^2K}{\alpha^2\tau^2} + \frac{8\eta^2 dp^2G^2K}{\tau^2} = \frac{8d\eta^2 p^2G^2K}{\tau^2}(1+\frac{2}{\alpha^2}) \ .\\
	\end{aligned}
	\end{equation}
	
\end{proof}

\begin{proof}
	From Eq.~(\ref{x_bar_compress}), it is easy to get 
	\begin{equation}
	\bar{\mathbf{x}}_{t+1} = 	\bar{\mathbf{x}}_{t} - \frac{\eta }{K}\sum_{k=1}^{K}\frac{\mathbf{g}_t^{(k)}}{\sqrt{\mathbf{v}_{t}^{(k)} } + \tau}  \ .
	\end{equation}
	Then, due to the smoothness of the loss function, we have
	\begin{equation} 
	\begin{aligned}
	& \mathbb{E}[f(\bar{\mathbf{x}}_{t+1})] \leq  \mathbb{E}[f(\bar{\mathbf{x}}_{t}) + \langle \nabla  f(\bar{\mathbf{x}}_{t}), \bar{\mathbf{x}}_{t+1}-\bar{\mathbf{x}}_{t} \rangle + \frac{L}{2} \|\bar{\mathbf{x}}_{t+1}-\bar{\mathbf{x}}_{t} \|^2] \\
	& =  f(\bar{\mathbf{x}}_{t}) - \mathbb{E}[\langle \nabla  f(\bar{\mathbf{x}}_{t}), \frac{\eta }{K}\sum_{k=1}^{K}\frac{\mathbf{g}_t^{(k)}}{\sqrt{\mathbf{v}_{t}^{(k)} } + \tau}  \rangle] + \frac{L}{2} \mathbb{E}[\|\frac{\eta }{K}\sum_{k=1}^{K}\frac{\mathbf{g}_t^{(k)}}{\sqrt{\mathbf{v}_{t}^{(k)} } + \tau} \|^2]  \ . \\
	\end{aligned}
	\end{equation}
	
	Similar as the proof of Theorem~\ref{theorem1}, we can get 
	\begin{equation}
	\begin{aligned}
	&  \mathbb{E}[f(\bar{\mathbf{x}}_{t+1}) ] \leq  f(\bar{\mathbf{x}}_{t}) -\frac{\eta }{2K} \sum_{k=1}^{K}  \sum_{j=1}^{d}\frac{[\nabla  f(\bar{\mathbf{x}}_{t})]_j^2}{\sqrt{\beta_2\mathbf{v}_{t-1, j}^{(k)} } + \tau} + \frac{\eta L^2}{2\tau K} \mathbb{E}[\sum_{k=1}^{K}  \|\mathbf{x}_{t}^{(k)} - \bar{\mathbf{x}}_{t}\|^2 ] \\
	& + \frac{\eta G\sqrt{1-\beta_2} }{\tau K}\mathbb{E}[\sum_{k=1}^{K} \sum_{j=1}^{d}  \frac{(\mathbf{g}_{t,j}^{(k)})^2}{\sqrt{\mathbf{v}_{t,j}^{(k)}}+\tau}]  + \frac{\eta^2L}{2K} \mathbb{E}[\sum_{k=1}^{K} \sum_{j=1}^{d} \frac{(\mathbf{g}_{t,j}^{(k)})^2}{\mathbf{v}_{t,j}^{(k)}+\tau^2} ] \ . \\
	\end{aligned}
	\end{equation}
	
	In terms of Lemma~\ref{lemma_x_diff_2}, we have
	\begin{equation}
	\begin{aligned}
	&  \mathbb{E}[f(\bar{\mathbf{x}}_{t+1}) ]   \leq  f(\bar{\mathbf{x}}_{t}) -\frac{\eta }{2K} \sum_{k=1}^{K}  \sum_{j=1}^{d}\frac{[\nabla  f(\bar{\mathbf{x}}_{t})]_j^2}{\sqrt{\beta_2\mathbf{v}_{t-1, j}^{(k)} } + \tau}  +\frac{4d\eta^3 p^2G^2L^2}{\tau^3}(1+\frac{2}{\alpha^2})\\
	&  + (\eta G\sqrt{1-\beta_2} + \frac{\eta^2L}{2})\Big(\frac{3}{\tau^2}\sum_{j=1}^d\sigma_{j}^2  + (1+ \frac{2}{\alpha^2}) \frac{24d\eta^2p^2G^2}{\tau^4}  + \frac{3}{\tau^2T} \sum_{t=0}^{T-1}  \|\nabla f(\bar{\mathbf{x}}_t) \|^2\Big)  \ . \\
	\end{aligned}
	\end{equation}
	%	where the second step follows from Lemma~\ref{lemma_x_diff_2}. 
	
	By setting $\eta<\frac{\tau^2}{3\sqrt{\beta_2}GL}$, we have
	\begin{equation}
	\begin{aligned}
	& \frac{1}{T}\sum_{t=0}^{T-1}\frac{\eta }{2K} \sum_{k=1}^{K}  \sum_{j=1}^{d}\frac{[\nabla  f(\bar{\mathbf{x}}_{t})]_j^2}{\sqrt{\beta_2\mathbf{v}_{t-1, j}^{(k)} } + \tau} \leq \frac{f(\mathbf{x}_0) -   f(\mathbf{x}_*) }{T}+  \frac{4d\eta^3 p^2G^2L^2}{\tau^3}(1+\frac{2}{\alpha^2})\\
	& + (\eta G\sqrt{1-\beta_2} + \frac{\eta^2L}{2})\Big(\frac{3}{\tau^2}\sum_{j=1}^d\sigma_{j}^2  + (1+ \frac{2}{\alpha^2}) \frac{24d\eta^2p^2G^2}{\tau^4}  \Big)  \ .\\
	\end{aligned}
	\end{equation}
	
	At last, by setting $0<\tau<1$ and $\alpha=\frac{\rho^2\delta}{82}$, we have  
	\begin{equation}
	\begin{aligned}
	& \frac{1}{T}\sum_{t=0}^{T-1} \|\nabla  f(\bar{\mathbf{x}}_{t})\|^2 \leq    (\sqrt{\beta_2}G + 1)\Bigg(\frac{f(\mathbf{x}_0) -   f(\mathbf{x}_*) }{\eta T}+  \frac{4d\eta^2 p^2G^2L^2}{\tau^3}(1+\frac{13448}{\rho^4\delta^2})\\
	& + ( G\sqrt{1-\beta_2} + \frac{\eta L}{2})\Big(\frac{3}{\tau^2}\sum_{j=1}^d\sigma_{j}^2 + (1+ \frac{13448}{\rho^4\delta^2}) \frac{24d\eta^2p^2G^2}{\tau^4}  \Big) \Bigg)  \ .\\
	\end{aligned}
	\end{equation}

\end{proof}

\subsection{Additional Lemmas}
\begin{lemma} \label{lemma_w_spectral}\cite{koloskova2019decentralizedcon}
	For the doubly stochastic matrix $W$ defined in Definition~\ref{graph}, we  have 
	\begin{equation}
	\|W-\frac{1}{K}\mathbf{1}\mathbf{1}^T\|_2\leq 1-\rho \  ,
	\end{equation}
	where $1-\rho=|\lambda_2|<1$. $\lambda_2$ is the second largest eigenvalue of $W$.
\end{lemma}